\documentclass[letterpaper, 10 pt, journal, twoside]{ieeetran}
\usepackage{amsmath,amsfonts,amssymb,amsthm}
\usepackage{algorithmic}
\usepackage{tabularx,xtab,array,float}
\usepackage{epsfig,epstopdf,wrapfig,psfrag}
\usepackage{tikz}
\usepackage{textcomp}
\usepackage{stfloats}
\usepackage{url}
\usepackage{verbatim}
\usepackage{graphicx}
\usepackage{physics,stackengine}
\usepackage{cite}
\def\BibTeX{{\rm B\kern-.05em{\sc i\kern-.025em b}\kern-.08em
    T\kern-.1667em\lower.7ex\hbox{E}\kern-.125emX}}
\usepackage{balance}
\usepackage{hyperref}

\makeatletter
\newcommand*{\encircled}[1]{\relax\ifmmode\mathpalette\@encircled@math{#1}\else\@encircled{#1}\fi}
\newcommand*{\@encircled@math}[2]{\@encircled{$\m@th#1#2$}}
\newcommand*{\@encircled}[1]{%
\tikz[baseline,anchor=base]{\node[draw,circle,outer sep=0pt,inner sep=.2ex] {#1};}}
\makeatother

\newtheorem{theorem}{Theorem}
\newtheorem{lemma}{Lemma}
\newtheorem{proposition}{Proposition}%
\newtheorem{corollary}{Corollary}%

\theoremstyle{remark}%
\newtheorem{remark}{Remark}%

\theoremstyle{definition}%
\newtheorem{definition}{Definition}%
\newtheorem{assumption}{Assumption}%

\newcommand{\qac}{h}
\newcommand{\dqac}{v}
\newcommand{\delths}{{\Delta\theta^*}}
\newcommand{\drho}{D\rho}
\newcommand{\Sp}{S_p}
\newcommand{\Sm}{S_m}

\begin{document}
\title{Discrete VHCs for Propeller Motion of a Devil-Stick using purely Impulsive Inputs}
\author{Aakash~Khandelwal, 
and~Ranjan~Mukherjee,~\IEEEmembership{Senior Member}
\thanks{This work was supported by the National Science Foundation, under Grant CMMI-2043464}%
\thanks{The authors are with the Department of Mechanical Engineering, Michigan State University, East Lansing, MI 48824, USA
{\tt\footnotesize khande10@egr.msu.edu}, {\tt\footnotesize mukherji@egr.msu.edu}}
}


\maketitle

\begin{abstract}
The control problem of realizing propeller motion of a devil-stick in the vertical plane using impulsive forces applied normal to the stick is considered. This problem is an example of underactuated robotic juggling and has not been considered in the literature before. Inspired by virtual holonomic constraints, the concept of discrete virtual holonomic constraints (DVHC) is introduced for the first time to solve this orbital stabilization problem. At the discrete instants when impulsive inputs are applied, the location of the center-of-mass of the devil-stick is specified in terms of its orientation angle. This yields the discrete zero dynamics (DZD), which provides conditions for stable propeller motion. In the limiting case, when the rotation angle between successive applications of impulsive inputs is chosen to be arbitrarily small, the problem reduces to that of propeller motion under continuous forcing. A controller that enforces the DVHC, and an orbit stabilizing controller based on the impulse controlled Poincar\'e map approach are presented. The efficacy of the approach to trajectory design and stabilization is validated through simulations.
\end{abstract}


\section*{Nomenclature}
\begin{tabularx}{\columnwidth}{lX}
$g$ & acceleration due to gravity, (m/s$^2$)\\
$h_x, h_y$ & Cartesian coordinates of the center-of-mass of the devil-stick, (m) \\
$\ell$ & length of the devil-stick, (m)\\
$m$ & mass of the devil-stick, (kg)\\
$r$ & distance of point of application of impulsive force from the center-of-mass of the devil-stick, it is considered to be positive when the moment of the impulsive force about the center-of-mass is counter-clockwise, (m) \\
$v_x, v_y$ & Cartesian components of velocity of the center-of-mass of the devil-stick, (m/s) \\
$I$ & impulse of impulsive force applied on the devil-stick, (Ns)  \\
$J$ & mass moment of inertia of the devil-stick about its center-of-mass, (kgm$^2$)\\
$\theta$ & orientation of the devil-stick, measured positive counter-clockwise with respect to the horizontal axis, (rad) \\
$\omega$ & angular velocity of the devil-stick, (rad/s) 
\end{tabularx}

\section{Introduction} \label{sec1}

The \emph{propeller} is a devil-stick trick where the stick is made to spin like a propeller about a horizontal axis. This paper addresses the control problem of realizing propeller motion of the devil-stick using intermittent impulsive forces, which has not been considered heretofore. The problem is an example of underactuated robotic juggling, where the objective is to achieve stable, periodic motion of the devil-stick in the vertical plane. Unlike previous work on robotic juggling, this problem cannot be reduced to one of fixed-point stabilization as the dynamics cannot be rendered invariant at successive configurations through a coordinate transformation.

The dynamics of juggling \cite{lynch_dynamic_1999, mason_progress_1999, ruggiero_nonprehensile_2018} is inherently hybrid, and involves a flow associated with the unconstrained motion of the object and a jump in states following impulsive actuation. The control problem for both point and extended objects have been considered. For point objects \cite{zavala-rio_control_1999, brogliato_control_2000, zavala-rio_direct_2001, spong_impact_2001, sepulchre_stabilization_2003}, a majority of the work has investigated the combined object-robot system, where the continuous-time dynamics of the actuator is modeled in addition to the dynamics of the object. The controllability and stability of such hybrid systems has been investigated and control algorithms for juggling have been proposed \cite{brogliato_controllability_2006, brogliato_modeling_2023, tornambe_modeling_1999, lynch_recurrence_2001, buehler_planning_1994, posa_stability_2016}.
One-DOF ball juggling was analyzed in \cite{zavala-rio_control_1999} using Poincar\'e maps that captured the combined dynamics of the impact and flight phases; the control problem of tracking a reference trajectory was realized using a hybrid controller. This procedure was extended to a class of complementary-slackness jugglers in \cite{brogliato_control_2000}, which modeled both the object and the robot connected by a unilateral constraint and an impact law based on the coefficient of restitution. A hybrid control strategy was used for trajectory tracking in \cite{sanfelice_hybrid_2007} and finite-time stabilization was achieved following impacts between the actuated platform and the ball. Blind jugglers have been studied in \cite{ronsse_rhythmic_2007} modeling the complete dynamics of the ball and actuator via a Poincar\'e map; a discrete reference for the ball was tracked by controlling the motion of the actuator as a function of impact times. A blind juggler for apex height control of a ball was realized in \cite{reist_design_2012}.

The problem of devil-stick juggling has been studied in \cite{schaal_open_1993, kant_non-prehensile_2021, kant_juggling_2022, khandelwal_nonprehensile_2023}. Compared to ball-juggling, this problem models the orientation of the stick, which introduces additional DOFs to be controlled. Planar stick juggling \cite{kant_non-prehensile_2021, kant_juggling_2022} is underactuated with a single degree of underactuation, while stick juggling in three dimensions \cite{khandelwal_nonprehensile_2023} has two degrees of underactuation, making the control problem more challenging.
In \cite{kant_non-prehensile_2021, kant_juggling_2022, khandelwal_nonprehensile_2023}, the hybrid dynamic model of the juggled stick is derived, and impulsive control inputs which stabilize the orbit associated with a desired steady juggling motion are obtained.

The above juggling problems of balls or sticks equate to the problem of stabilizing a fixed point (which may be time-varying – see \cite{khandelwal_maneuvering_2024}, for example) of a Poincar\'e map describing the dynamics of the juggled object in an appropriate reference frame. To address more complex underactuated juggling tasks, it would be preferable to avoid specifying a reference trajectory for the juggled object. To this end, this paper extends the notion of virtual holonomic constraints (VHCs) \cite{maggiore_virtual_2013, mohammadi_dynamic_2018}, which avoid the need for the controller to track time-varying reference trajectories.

VHCs have been used extensively for trajectory design and stabilization for underactuated systems with continuous-time dynamics such as the cart-pendulum \cite{shiriaev_constructive_2005}, as well as hybrid dynamics such as bipedal robots \cite{grizzle_asymptotically_2001, kao-vukovich_synthesis_2023, khandelwal_design_2023}. Once the VHCs are enforced using an appropriate control design, the stability properties of the system can be deduced from investigating the induced zero dynamics. This has been done extensively for bipeds, in which the stability of the gait is governed by the hybrid zero dynamics (HZD). VHCs have also been used for trajectory design in nonprehensile manipulation tasks involving continuous contact between the object and actuator \cite{surov_case_2015}; this includes the problem of propeller motion of a devil-stick \cite{shiriaev_generating_2006, khandelwal_propeller_2025}. 

Orbital stabilization of propeller motion under continuous forcing has been considered in \cite{kawaida_feedback_2003, nakaura_enduring_2004, nakamura_enduring_2009, shiriaev_generating_2006, aoyama_realization_2015} using normal and tangential forcing. These approaches assumed the actuator rolling without slipping on the devil stick, and instantaneous resetting of the contact point following a complete rotation of the devil-stick. In \cite{khandelwal_propeller_2025}, the authors presented an alternative approach that omitted tangential forcing, and controlled the applied normal force and its point of application instead. A circular trajectory for the center-of-mass of the devil stick was specified using VHCs and it was shown that a desired propeller motion could be stabilized from arbitrary initial conditions.

This paper considers the problem of realizing propeller motion of a devil-stick using purely impulsive forces applied normal to the stick. Unlike \cite{kant_non-prehensile_2021, khandelwal_nonprehensile_2023}, this orbital stabilization problem is not reducible to one of fixed-point stabilization in a rotating reference frame. To solve this problem, we introduce the concept of discrete virtual holonomic constraints (DVHCs). Similar to VHCs, DVHCs are geometric constraints on the system coordinates that eliminate the need for tracking time-varying reference trajectories. However, while VHCs are enforced by continuous feedback such that they are satisfied for all time, DVHCs are enforced by discrete feedback and are satisfied at the discrete instants of time when the system is subject to impulsive control inputs, but not necessarily in the interval between successive impulsive inputs. We specify a trajectory for the devil-stick such that its center-of-mass lies on a circle at the discrete instants when impulsive forces are applied on the stick. These instants correspond to rotation of the devil-stick by a constant angle, which is chosen \emph{a priori}.
A control design that enforces the DVHC from arbitrary initial conditions is presented, and it is shown that it drives system trajectories to a solution set which is a discrete-time analog of the constraint manifold in continuous-time. As with VHCs, the behavior of the dynamical system resulting from enforcing the DVHCs is described by autonomous equations in the passive coordinate and velocity alone – the \emph{discrete zero dynamics} (DZD). It is shown that the DZD permits a collection of periodic orbits which are each stable, but not asymptotically stable. We present a control design employing the impulse controlled Poincar\'e map (ICPM) approach \cite{kant_orbital_2020}, which has been applied to both continuous and hybrid systems \cite{kant_juggling_2022, khandelwal_nonprehensile_2023, khandelwal_design_2023}, to asymptotically stabilize a specific periodic orbit.\ 

The main contributions of this work are:
\begin{itemize}
    \item Introducing discrete virtual holonomic constraints (DVHCs) as a tool for trajectory design in underactuated systems which permit application of intermittent impulsive control inputs only.
    \item Presenting a control strategy to enforce the DVHC for an underactuated devil-stick juggling task.
    \item Deriving the discrete zero dynamics (DZD) induced from applying the DVHC, and investigating its stability properties through analysis and simulation. 
    \item Applying the ICPM approach to stabilize the orbit defining a specific juggling motion from a collection of orbits.
    \item Demonstrating that the constraints on velocities induced by DVHCs and the DZD presented here reduce to the established results for continuous-time propeller motion in the limit when the rotation angle between successive applications of impulsive inputs is arbitrarily small.
\end{itemize}

This paper is organized as follows. In Section \ref{sec2}, we derive the hybrid dynamics of the system. We introduce the concept of DVHCs in Section \ref{sec3}, compare them with VHCs, and present a control design that enforces the DVHCs. In section \ref{sec4}, we derive the DZD, which describes the qualitative behavior of the system when the DVHCs are satisfied. The DZD is investigated in detail in Appendix \ref{appendix:zero-dyn}. In Section \ref{sec5}, we describe a process to stabilize a desired propeller motion. Simulation results are presented in Section \ref{sec6}, and demonstrate the effectiveness of the approach for trajectory design and stabilization.\

\section{System Dynamics} \label{sec2}

\subsection{System Description} \label{sec21}

\begin{figure}[b!]
    \centering
    \psfrag{A}{$x$}
    \psfrag{B}{$y$}
    \psfrag{C}{$\ell$}
    \psfrag{D}{$\theta$}
    \psfrag{E}{$m,\, J$}
    \psfrag{F}{\small{center-of-mass $G$}}
    \psfrag{K}{$g$}
    \psfrag{P}{$I$}
    \psfrag{Q}{$r$}
    \includegraphics[width=0.58\linewidth]{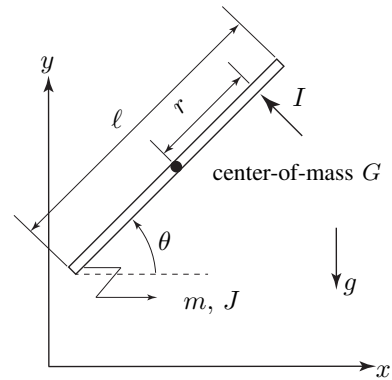}
    \caption{A devil-stick in the vertical plane with configuration variables $(h_x, h_y, \theta)$, and control variables $(I, r)$.}
    \label{fig:description}
\end{figure}

Consider the devil-stick in the vertical $xy$ plane as shown in Fig.\ref{fig:description}. The devil-stick has three DOFs and its configuration is described by the three independent generalized coordinates $(h_x, h_y, \theta)$.
The devil-stick is manipulated using purely impulsive forces of impulse $I$ applied normal to the devil-stick at a distance $r$ from the center-of-mass; $I$ and $r$ are considered to be the control inputs for the system. The impulsive force is applied only when the coordinate $\theta$ evolves by a constant amount $\delths$, $\delths \in (0, \pi)$, which is chosen \emph{a priori}. 
Between successive applications of the impulsive force, the devil-stick falls freely under gravity. The resulting dynamics of the system is therefore hybrid.
Since the number of control inputs is one fewer than the number of generalized coordinates, the system is underactuated with a single passive DOF. The vector of generalized coordinates is represented by
\begin{equation} \label{eq:generalized-coord}
    q \triangleq \begin{bmatrix} h^T & \!\!|\,\,\theta \end{bmatrix}^T, \quad h = \begin{bmatrix} h_x & h_y \end{bmatrix}^T
\end{equation}

\noindent where $h \in \mathbb{R}^2$ are treated as the active coordinates and $\theta \in \mathbb{R}$ is treated as the passive coordinate. The vector of generalized velocities is denoted by
\begin{equation} \label{eq:generalized-vel}
    \dot q \triangleq \begin{bmatrix} v^T & \!\!|\,\,\omega \end{bmatrix}^T, \quad v = \begin{bmatrix} v_x & v_y \end{bmatrix}^T
\end{equation}

\noindent where $v \in \mathbb{R}^2$ and $\omega \in \mathbb{R}$. This work will focus on realizing propeller motion of the devil-stick in the counterclockwise direction, and assumes that $\omega > 0$.\

\subsection{Dynamic Model} \label{sec22}

\subsubsection{Impulsive Dynamics} \label{sec221}

Let $t_k, k = 1, 2, \dots,$ denote the instants when impulsive inputs are applied on the stick, with $t_k^-$ and $t_k^+$ denoting instants immediately before and after application of the impulsive inputs. Since impulsive inputs cause no change in position coordinates \cite{kant_orbital_2020}, 
\begin{equation} \label{eq:impulsive-posn}
    \begin{split}
        h_x(t_k^+) &= h_x(t_k^-), \quad h_y(t_k^+) = h_y(t_k^-) \\
        \theta(t_k^+) &= \theta(t_k^-) \triangleq \theta_k
    \end{split}
\end{equation}

\noindent where $\theta_k$ is defined for notational simplicity. The linear impulse-momentum relationships in the $x$ and $y$ directions give
\begin{subequations}
\begin{align}
    m v_x(t_k^+) &= m v_x(t_k^-) - I_k \sin\theta_k \\
    m v_y(t_k^+) &= m v_y(t_k^-) + I_k \cos\theta_k  
\end{align}
\end{subequations}

\noindent where $I_k \triangleq I(t_k)$. It follows that the discontinuous jumps in the Cartesian velocities are
\begin{subequations}
\begin{align}
    v_x(t_k^+) &= v_x(t_k^-) - \frac{I_k}{m} \sin\theta_k \label{eq:impulsive-vel-x} \\
    v_y(t_k^+) &= v_y(t_k^-) + \frac{I_k}{m} \cos\theta_k \label{eq:impulsive-vel-y}
\end{align} 
\end{subequations}

\noindent The angular impulse-momentum relationship gives
\begin{align}
    J \omega(t_k^+) &= J \omega(t_k^-) + I_k r_k \\
    \Rightarrow \qquad \omega(t_k^+) &= \omega(t_k^-) + \frac{I_k r_k}{J} \label{eq:impulsive-ang-vel}
\end{align}

\noindent where $r_k \triangleq r(t_k)$.\

\subsubsection{Continuous-time Dynamics} \label{sec222}

Between successive impulsive inputs, \emph{i.e.}, in the interval $t \in [t_k^+, t_{k+1}^-]$, the devil-stick undergoes torque-free motion under gravity. The motion of the center-of-mass is described by:
\begin{alignat}{2}
    \dot h_x &= v_x, \quad \dot h_y &&= v_y \label{eq:continuous-time-pos} \\
    \dot v_x &= 0, \quad\,\,\, \dot v_y &&= -g \label{eq:continuous-time-vel}
\end{alignat}

\noindent Using $h_x(t_k^+)$, $h_y(t_k^+)$, $v_x(t_k^+)$, and $v_y(t_k^+)$ from \eqref{eq:impulsive-posn}, \eqref{eq:impulsive-vel-x}, and \eqref{eq:impulsive-vel-y} as initial conditions, the solutions to \eqref{eq:continuous-time-pos} and \eqref{eq:continuous-time-vel} are obtained as
\begin{subequations} \label{eq:pos-soln}
\begin{align}
    h_x(t_{k+1}^-) &= h_x(t_k^-) + v_x(t_k^-) \delta_k - \frac{I_k}{m} \sin\theta_k \delta_k \\
    h_y(t_{k+1}^-) &= h_y(t_k^-) + v_y(t_k^-) \delta_k + \frac{I_k}{m} \cos\theta_k \delta_k - \frac{1}{2}g \delta_k^2
\end{align}
\end{subequations}
\begin{subequations} \label{eq:vel-soln}
\begin{align}
    v_x(t_{k+1}^-) &= v_x(t_k^-) - \frac{I_k}{m} \sin\theta_k \label{eq:vel-soln-x} \\
    v_y(t_{k+1}^-) &= v_y(t_k^-) + \frac{I_k}{m} \cos\theta_k - g \delta_k \label{eq:vel-soln-y}
\end{align}
\end{subequations}

\noindent where the time-of-flight $\delta_k \triangleq (t_{k+1}^- - t_k^-)$ is the interval between application of the $k$-th and $(k+1)$-th impulsive inputs. Since angular momentum is conserved, we have
\begin{equation}
    \dot\theta = \omega, \quad J \dot\omega = 0 \quad \Rightarrow \quad \dot\omega = 0
\end{equation}

\noindent Using $\theta_k$ and $\omega(t_k^+)$ from \eqref{eq:impulsive-posn} and \eqref{eq:impulsive-ang-vel} as initial conditions, the above equations give
\begin{align}
    \theta_{k+1} &= \theta_k +  \omega(t_k^-) \delta_k + \frac{I_k r_k}{J} \delta_k \label{eq:theta-sol} \\
    \omega(t_{k+1}^-) &= \omega(t_k^-) + \frac{I_k r_k}{J} \label{eq:omega-sol}
\end{align}

\subsection{Hybrid Dynamic Model} \label{sec23}

The hybrid dynamic model captures the dynamics of the system in the interval $t \in [t_k^-, t_{k+1}^-]$ and encompasses a single impulsive actuation at time $t_k$. Using $k$ in place of $t_k^-$ to represent instants immediately prior to application of impulsive inputs, \eqref{eq:pos-soln} and \eqref{eq:vel-soln} can be rewritten as
\begin{equation} \label{eq:hybrid-pos}
    \qac(k+1) = \qac(k) + \dqac(k) \delta_k 
    + \begin{bmatrix} -\sin\theta_k \\ \cos\theta_k \end{bmatrix} \frac{I_k \delta_k}{m}
    + \begin{bmatrix} 0 \\ - \frac{1}{2}g \delta_k^2 \end{bmatrix}
\end{equation}
\begin{equation} \label{eq:hybrid-vel}
    \dqac(k+1) = \dqac(k)
    + \begin{bmatrix} -\sin\theta_k \\ \cos\theta_k \end{bmatrix} \frac{I_k}{m}
    + \begin{bmatrix} 0 \\ - g \delta_k \end{bmatrix}
\end{equation}

\noindent Since impulsive inputs are only applied when
\begin{equation} \label{eq:hybrid-theta-delta}
    \theta_{k+1} - \theta_k = \delths
\end{equation}

\noindent equations \eqref{eq:theta-sol} and \eqref{eq:omega-sol} may be rewritten as
\begin{align}
    \delths &= \omega_k \delta_k + \frac{I_k r_k}{J} \delta_k \label{eq:hybrid-theta} \\
    \omega_{k+1} &= \omega_k + \frac{I_k r_k}{J} \label{eq:hybrid-omega}
\end{align}

\noindent It follows from \eqref{eq:hybrid-theta} and \eqref{eq:hybrid-omega} that the time-of-flight is
\begin{equation} \label{eq:time-of-flight}
    \delta_k = \frac{\delths}{\omega_k + \dfrac{I_k r_k}{J}} \quad \Rightarrow \quad \delta_k = \frac{\delths}{\omega_{k+1}}
\end{equation}

\noindent The hybrid dynamics of the system is completely described by \eqref{eq:hybrid-pos}, \eqref{eq:hybrid-vel}, \eqref{eq:hybrid-theta-delta}, and \eqref{eq:hybrid-omega}, with $\delta_k$ given by \eqref{eq:time-of-flight}.

\section{Discrete Virtual Holonomic Constraints} \label{sec3}

\subsection{Virtual Holonomic Constraints at Discrete Instants} \label{sec31}

\begin{figure}[b!]
    \centering
    \psfrag{A}{$x$}
    \psfrag{B}{$y$}
    \psfrag{C}{$R$}
    \psfrag{D}{$\theta_k$}
    \psfrag{E}{$\phi$}
    \psfrag{F}{\small{$\theta_{k+1}$}}
    \psfrag{K}{$g$}
    \psfrag{P}{$I$}
    \psfrag{Q}{$r$}
    \includegraphics[width=0.58\linewidth]{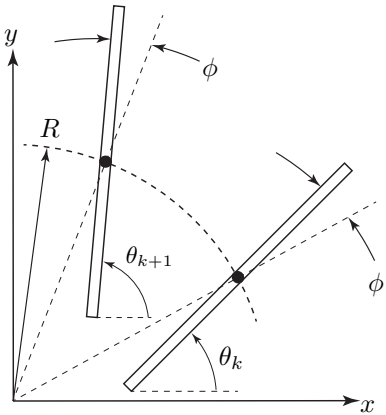}
    \caption{DVHC describing the desired trajectory of the center-of-mass of the devil-stick at $t = t_k$ and $t = t_{k+1}$ showing the parameters $R$ and $\phi$.}
    \label{fig:vhc}
\end{figure}

In contrast to VHCs, which are relations in the coordinates $q$ made invariant by feedback in continuous-time, we introduce the notion of discrete virtual holonomic constraints.
\begin{definition}
\label{DVHC}
     A \emph{discrete virtual holonomic constraint} (DVHC) is a relation in the coordinates $q$ enforced by discrete feedback such that the relation is satisfied at the instants $t = t_k$, $k = 1, 2, \dots$, but not necessarily for $t \in [t_k^+, t_{k+1}^-]$.
\end{definition}

We use Definition \ref{DVHC} to define the trajectory of the center-of-mass of the devil-stick in terms of its orientation:
\begin{equation} \label{eq:VHC-rho}
    \rho_k \triangleq \qac(k) - \Phi(\theta_k) = 0, \quad \Phi: \mathbb{R} \rightarrow \mathbb{R}^2, \ k = 1,2, \dots 
\end{equation}

\noindent where $\Phi: \mathbb{R} \rightarrow \mathbb{R}^2$ is given by 
\begin{equation} \label{eq:VHC-Phi}
    \Phi(\theta_k) = \begin{bmatrix}
        R \cos(\theta_k - \phi) \\
        R \sin(\theta_k - \phi)
    \end{bmatrix}
\end{equation}

\noindent Physically, this DVHC ensures that the center-of-mass of the devil-stick lies on a circle of radius $R$, $R > 0$, centered at the origin of the $xy$ frame, at the instants $t = t_k$, $k = 1, 2, \dots$, as shown in Fig.\ref{fig:vhc}. The parameter $\phi$, $\phi \in (-\pi, \pi]$, denotes an offset between the angle subtended by the center-of-mass $G$ at the origin, and the orientation $\theta$. Note that the function $\Phi$ in \eqref{eq:VHC-Phi} is identical to the VHC in \cite{khandelwal_propeller_2025}, used to define the trajectory for continuous-time propeller motion of a devil-stick.\

\begin{remark}
    Since the DVHC is only satisfied at $t = t_k$, $k = 1, 2, \dots$, is assumed that the system does not have finite escape times in the interval $t \in [t_k^+, t_{k+1}^-]$.
    This assumption holds for the system considered in this paper, as the devil-stick falls freely under gravity when $t \in [t_k^+, t_{k+1}^-]$ - see Section \ref{sec222}.
\end{remark}

\begin{remark} \label{rem-vhc}
    A VHC defines a \emph{constraint manifold} \cite{maggiore_virtual_2013, kant_orbital_2020} on which the tangent space is defined by the velocities $\dqac = \dot \Phi$ for all time; the constraints on the active velocities are therefore independent of the system dynamics. The same is not true for a DVHC; knowledge of the system dynamics is necessary to derive the constraints on the active velocities at $t = t_k^-$ that correspond to the DVHC.
\end{remark}

We now derive the constraints on the active velocities corresponding to the DVHC in \eqref{eq:VHC-rho} and \eqref{eq:VHC-Phi}, based on knowledge of the hybrid dynamics from Section \ref{sec23}.

\subsection{Constraints on Velocities} \label{sec32}

We begin by rewriting \eqref{eq:hybrid-pos} as
\begin{equation}
\begin{split}
    \qac(k+1) - \qac(k) = \left\{\dqac(k)
    + \begin{bmatrix} -\sin\theta_k \\ \cos\theta_k \end{bmatrix} \frac{I_k}{m} \right. \\ \left. 
    + \begin{bmatrix} 0 \\ - g \delta_k \end{bmatrix}\right\} \delta_k
    + \begin{bmatrix} 0 \\ \frac{1}{2}g \delta_k^2 \end{bmatrix}
\end{split}
\end{equation}

\noindent Using \eqref{eq:hybrid-vel} in the above equation, we obtain
\begin{equation} \label{eq:hybrid-pos-vel}
    \qac(k+1) - \qac(k) = \dqac(k+1) \delta_k 
    + \begin{bmatrix} 0 \\ \frac{1}{2}g \delta_k^2 \end{bmatrix}
\end{equation}

\noindent Now, using the DVHC from \eqref{eq:VHC-rho} at $k$ and $k+1$ in the above equation, we obtain
\begin{equation}
    \Phi(\theta_{k+1}) - \Phi(\theta_k) = \dqac(k+1) \delta_k 
    + \begin{bmatrix} 0 \\ \frac{1}{2}g \delta_k^2 \end{bmatrix}
\end{equation}

\noindent It follows from \eqref{eq:hybrid-theta-delta} that $\theta_k = \theta_{k+1} - \delths$, which can be used to rewrite the above equation as
\begin{align}
    \Phi(\theta_{k+1}) - \Phi(\theta_{k+1} - \delths) = \dqac(k+1) \delta_k 
    &+ \begin{bmatrix} 0 \\ \frac{1}{2}g \delta_k^2 \end{bmatrix} \notag \\ \Rightarrow
    \dqac(k+1) = \frac{1}{\delta_k} [\Phi(\theta_{k+1}) - \Phi(\theta_{k+1} - \delths)] 
    &- \begin{bmatrix} 0 \\ \frac{1}{2}g \delta_k \end{bmatrix} \label{eq:vhc-vel-time-of-flight}
\end{align}

\noindent Using \eqref{eq:time-of-flight}, the above equation can be rewritten as
\begin{equation} 
    \dqac(k+1) = \frac{\omega_{k+1}}{\delths} [\Phi(\theta_{k+1}) - \Phi(\theta_{k+1} - \delths)] 
    - \begin{bmatrix} 0 \\ \dfrac{g \delths}{2 \omega_{k+1}} \end{bmatrix}
\end{equation}

\noindent Replacing all instances of $k+1$ by $k$ in the above equation, we obtain
\begin{equation} 
    \dqac(k) = \frac{\omega_{k}}{\delths} [\Phi(\theta_k) - \Phi(\theta_k - \delths)] 
    - \begin{bmatrix} 0 \\ \dfrac{g \delths}{2 \omega_k} \end{bmatrix}
\end{equation}

\noindent The constraint on the active velocities at the instants $t = t_k^-$ can therefore be described by the relation
\begin{equation} \label{eq:VHC-Drho}
    \drho_k \triangleq \dqac(k) - \Psi(\theta_k, \omega_k) = 0, \quad k = 1,2, \dots 
\end{equation}

\noindent where $\Psi$ is given by
\begin{equation} \label{eq:VHC-Psi-Phi}
    \Psi(\theta_k, \omega_k) = 
    \frac{\omega_{k}}{\delths} [\Phi(\theta_k) - \Phi(\theta_k - \delths)] 
    - \begin{bmatrix} 0 \\ \dfrac{g \delths}{2 \omega_k} \end{bmatrix}
\end{equation}

\noindent Using \eqref{eq:VHC-Phi} and the trigonometric identities for the differences of cosines and sines, the above equation may be rewritten as
\begin{equation} \label{eq:VHC-Psi}
\begin{split}
    \Psi(\theta_k, \omega_k) = \frac{\omega_{k}}{\delths} &\begin{bmatrix}
    - 2 R \sin\left(\dfrac{\delths}{2}\right) \sin\left(\theta_k - \phi - \dfrac{\delths}{2}\right) \\
    2 R \sin\left(\dfrac{\delths}{2}\right) \cos\left(\theta_k - \phi - \dfrac{\delths}{2}\right) 
    \end{bmatrix} \\
    - &\begin{bmatrix} 0 \\ \dfrac{g \delths}{2 \omega_k} \end{bmatrix}
\end{split}
\end{equation}

\noindent Similar to a VHC, the above constraint on the active velocities corresponding to the DVHC depends on the passive coordinate $\theta$ and the passive velocity $\omega$. 
It can be verified from \eqref{eq:VHC-Psi-Phi} that
\begin{equation}
    \lim_{\delths \rightarrow 0} \Psi = \left[\frac{\partial \Phi}{\partial \theta}\right]\omega = \dot\Phi
\end{equation}

\noindent which shows that the constraint on the active velocities reduces to the equivalent expression for VHCs in the limit $\delths \to 0$.

\subsection{Control Design Enforcing the DVHC} \label{sec33}

We now solve for the control inputs $I_k$ and $r_k$ that drive system trajectories exponentially to $\rho_k = 0$, 
first considering $\delta_k$ as an input for simplicity. Using \eqref{eq:VHC-rho} and \eqref{eq:VHC-Drho} in \eqref{eq:hybrid-pos}, we obtain
\begin{equation} \label{eq:hybrid-pos-rho}
\begin{split}
    \rho_{k+1} + \Phi(\theta_{k+1}) = &[\rho_k + \Phi(\theta_k)] + [\drho_k + \Psi(\theta_k, \omega_k)] \delta_k \\
    &+ \begin{bmatrix} -\sin\theta_k \\ \cos\theta_k \end{bmatrix} \frac{I_k \delta_k}{m}
    + \begin{bmatrix} 0 \\ -\frac{1}{2} g \delta_k^2 \end{bmatrix}
\end{split}
\end{equation}

\noindent To ensure that $\rho_k \to 0$ as $k \to \infty$ exponentially, we choose
\begin{equation} \label{eq:rho-lambda}
    \rho_{k+1} = \lambda \rho_k, \quad \lambda \triangleq \mathrm{diag}\begin{bmatrix}
        \lambda_x & \lambda_y
    \end{bmatrix},\, \lambda_x, \lambda_y \in [0,1)
\end{equation}

\noindent Note from \eqref{eq:hybrid-theta-delta} that $\Phi(\theta_{k+1}) = \Phi(\theta_k + \delths)$, which is known at $k$. We define $\eta_k \triangleq \Phi(\theta_k + \delths) - \Phi(\theta_k)$ and rewrite \eqref{eq:hybrid-pos-rho} as
\begin{equation} \label{eq:control-rho}
\begin{split}
    (\lambda - 1) \rho_k + \eta_k = &[\drho_k + \Psi(\theta_k, \omega_k)] \delta_k \\
    &+ \begin{bmatrix} -\sin\theta_k \\ \cos\theta_k \end{bmatrix} \frac{I_k \delta_k}{m}
    + \begin{bmatrix} 0 \\ -\frac{1}{2} g \delta_k^2 \end{bmatrix}
\end{split}
\end{equation}

\noindent The $x$ and $y$ components of the above equation can be expressed separately as
\begin{subequations} \label{eq:control-rho-xy}
\begin{align}
    (\lambda_x-1) \rho_{xk} + \eta_{xk} &- [\drho_{xk} + \Psi_x(\theta_k, \omega_k)] \delta_k \notag \\
    &= - \frac{I_k \delta_k}{m} \sin\theta_k \label{eq:control-rho-x} \\
    (\lambda_y-1) \rho_{yk} + \eta_{yk} &- [\drho_{yk} + \Psi_y(\theta_k, \omega_k)] \delta_k + \frac{1}{2} g \delta_k^2 \notag \\
    &= \frac{I_k \delta_k}{m} \cos\theta_k \label{eq:control-rho-y}
\end{align}
\end{subequations}

\noindent where $\rho_k = \begin{bmatrix} \rho_{xk} & \rho_{yk} \end{bmatrix}^T$, $\drho_k = \begin{bmatrix} \drho_{xk} & \drho_{yk} \end{bmatrix}^T$, $\eta_k = \begin{bmatrix} \eta_{xk} & \eta_{yk} \end{bmatrix}^T$ and $\Psi = \begin{bmatrix} \Psi_{x} & \Psi_{y} \end{bmatrix}^T$. When $\sin\theta_k = 0$, \eqref{eq:control-rho-x} can be solved for $\delta_k$ as
\begin{subequations} \label{eq:delta-k}
\begin{equation} \label{eq:delta-k-sin-thk-zero}
    \delta_k = \frac{(\lambda_x-1) \rho_{xk} + \eta_{xk}}{\drho_{xk} + \Psi_x(\theta_k, \omega_k)}
\end{equation}

\noindent When $\sin\theta_k \neq 0$, we can eliminate $I_k$ between \eqref{eq:control-rho-x} and \eqref{eq:control-rho-y} to obtain the following quadratic equation in $\delta_k$:
\begin{equation} \label{eq:delta-k-quad}
\begin{split}
    \frac{1}{2} g \delta_k^2 &- \left\{[\drho_{xk} + \Psi_x(\theta_k, \omega_k)] \cot\theta_k \right. \\
    &\left. + [\drho_{yk} + \Psi_y(\theta_k, \omega_k)] \right\} \delta_k 
    + \left[\eta_{xk} \cot\theta_k + \eta_{yk} \right. \\
    &\left.+ (\lambda_x-1) \rho_{xk} \cot\theta_k + (\lambda_y-1) \rho_{yk} \right] = 0
\end{split}
\end{equation}
\end{subequations}

\noindent which can be solved for $\delta_k$. When $\theta_k \!\!\mod 2\pi \in (0, \pi)$, we choose the `negative' square root solution to the quadratic, and when $\theta_k \!\!\mod 2\pi \in (\pi, 2\pi)$, we choose the `positive' square root solution to the quadratic\footnote{This result is consistent across extensive simulations; it ensures that the value of $\delta_k > 0$ and that it is closer in magnitude to $\delta_{k-1}$.}. 
Having computed $\delta_k$, we use \eqref{eq:control-rho-xy} to solve for the value of $I_k$. When $\sin\theta_k = 0$, we have from \eqref{eq:control-rho-y}:
\begin{subequations} \label{eq:Ik}
\begin{equation}
    I_k = \frac{m \left\{ (\lambda_y-1) \rho_{yk} + \eta_{yk} - [\drho_{yk} + \Psi_y(\theta_k, \omega_k)] \delta_k + \frac{1}{2} g \delta_k^2 \right\}}{\delta_k \cos\theta_k}
\end{equation}

\noindent and when $\sin\theta_k \neq 0$, we have from \eqref{eq:control-rho-x}:
\begin{equation}
    I_k = -\frac{m \left\{ (\lambda_x-1) \rho_{xk} + \eta_{xk} - [\drho_{xk} + \Psi_x(\theta_k, \omega_k)] \delta_k \right\}}{\delta_k \sin\theta_k}
\end{equation}
\end{subequations}

\noindent Finally, having obtained $\delta_k$ and $I_k$, we can find the value of $r_k$ using \eqref{eq:time-of-flight} as
\begin{equation} \label{eq:rk}
    r_k = \frac{J \delths}{I_k \delta_k} - \frac{J \omega_k}{I_k}
\end{equation}

\begin{remark}
    It must be ensured that the time-of-flight $\delta_k > 0$, and the point of application of the impulse lies on the devil-stick, \emph{i.e.}, $r_k \in (-\ell/2, \ell/2)$. 
\end{remark}

\noindent The control inputs from \eqref{eq:Ik} and \eqref{eq:rk} drive the system trajectories to $\rho_k = 0$. We now show that as $\rho_k \to 0$, $\drho_k \to 0$. Using \eqref{eq:VHC-rho} and \eqref{eq:VHC-Drho}, \eqref{eq:hybrid-pos-vel} can be rewritten 
\begin{equation}
\begin{split}
    &\rho_{k+1} - \rho_k + \Phi(\theta_{k+1}) - \Phi(\theta_k) \\
    = [&\drho_{k+1} + \Psi(\theta_{k+1}, \omega_{k+1})] \delta_k 
    + \begin{bmatrix} 0 \\ \frac{1}{2}g \delta_k^2 \end{bmatrix}
\end{split}
\end{equation}

\noindent Dividing both sides by $\delta_k$,
\begin{equation} \label{eq:Drho-control}
\begin{split}
    &\frac{\rho_{k+1} - \rho_k}{\delta_k} + \frac{\Phi(\theta_{k+1}) - \Phi(\theta_k)}{\delta_k} \\
    = &\drho_{k+1} + \Psi(\theta_{k+1}, \omega_{k+1}) 
    + \begin{bmatrix} 0 \\ \frac{1}{2}g \delta_k \end{bmatrix}
\end{split}
\end{equation}

\noindent It follows from \eqref{eq:VHC-Psi-Phi} that
\begin{equation*}
    \Psi(\theta_{k+1}, \omega_{k+1}) = 
    \frac{\omega_{k+1}}{\delths} [\Phi(\theta_{k+1}) - \Phi(\theta_k)] 
    - \begin{bmatrix} 0 \\ \dfrac{g \delths}{2 \omega_{k+1}} \end{bmatrix}
\end{equation*}

\noindent which, using \eqref{eq:time-of-flight} can be rewritten
\begin{equation}
    \Psi(\theta_{k+1}, \omega_{k+1}) = \frac{\Phi(\theta_{k+1}) - \Phi(\theta_k)}{\delta_k}
    - \begin{bmatrix} 0 \\ \frac{1}{2}g \delta_k \end{bmatrix}
\end{equation}

\noindent Using the above equation in \eqref{eq:Drho-control}, we obtain
\begin{equation} \label{eq:drho-lambda}
    \drho_{k+1} = \frac{\rho_{k+1} - \rho_k}{\delta_k} \quad \Rightarrow \quad \drho_{k+1} = \frac{(\lambda - 1)\rho_k}{\delta_k}
\end{equation}

\noindent which shows that $\drho_k \to 0$ as $\rho_k \to 0$, and the control design presented here enforces the DVHC.

\begin{remark}
    If $\lambda_x = \lambda_y = 0$, it follows from \eqref{eq:rho-lambda} that the control design presented in this section achieves dead-beat convergence to $\rho_k = 0$; consequently, it follows from \eqref{eq:drho-lambda} that $\drho_k$ is driven to zero in two steps.
\end{remark}

\section{Discrete Zero Dynamics} \label{sec4}

We now seek the dynamics of the passive coordinate $\theta$ and passive velocity $\omega$ when the DVHCs are satisfied and the system trajectories evolve such that $\rho_k \equiv 0$ and $\drho_k \equiv 0$:
\begin{equation} \label{eq:h-v-zero-dyn}
    \qac (k) \equiv \Phi(\theta_k), \quad \dqac (k) \equiv \Psi(\theta_k, \omega_k)
\end{equation}

\noindent From \eqref{eq:VHC-Psi}, it follows
\begin{subequations} \label{eq:VHC-Psi-k}
\begin{align}
    v_x (k) &= - \frac{2 R \omega_k}{\delths} \sin\left(\frac{\delths}{2}\right) \sin\left(\theta_k - \phi - \frac{\delths}{2}\right) \\
    v_y (k) &= \frac{2 R \omega_k}{\delths} \sin\left(\frac{\delths}{2}\right) \cos\left(\theta_k - \phi - \frac{\delths}{2}\right) - \frac{g \delths}{2 \omega_k}
\end{align}
\end{subequations}

\noindent Similarly, using \eqref{eq:hybrid-theta-delta} in \eqref{eq:VHC-Psi}, it follows
\begin{subequations} \label{eq:VHC-Psi-kp1}
\begin{align}
    v_x (k+1) &= - \frac{2 R \omega_{k+1}}{\delths} \sin\left(\frac{\delths}{2}\right) \sin\left(\theta_k - \phi + \frac{\delths}{2}\right) \\
    v_y (k+1) &= \frac{2 R \omega_{k+1}}{\delths} \sin\left(\frac{\delths}{2}\right) \cos\left(\theta_k - \phi + \frac{\delths}{2}\right) \notag \\ 
    & \quad\  - \frac{g \delths}{2 \omega_{k+1}}
\end{align}
\end{subequations}

\noindent Multiplying both sides of \eqref{eq:vel-soln-x} by $\cos\theta_k$, both sides of \eqref{eq:vel-soln-y} by $\sin\theta_k$, and adding the resulting equations we can eliminate the input $I_k$ between them to obtain
\begin{equation} 
\begin{split}
    &v_x(k+1) \cos\theta_k + v_y(k+1) \sin\theta_k \\= &v_x(k) \cos\theta_k + v_y(k) \sin\theta_k - g \delta_k \sin\theta_k
\end{split}
\end{equation}

\noindent Using \eqref{eq:time-of-flight}, the above equation may be rewritten as
\begin{equation} \label{eq:eliminate-Ik}
\begin{split}
    &v_x(k+1) \cos\theta_k + v_y(k+1) \sin\theta_k \\= &v_x(k) \cos\theta_k + v_y(k) \sin\theta_k - \frac{g \delths}{\omega_{k+1}} \sin\theta_k
\end{split}
\end{equation}

\noindent Using \eqref{eq:VHC-Psi-k} and \eqref{eq:VHC-Psi-kp1} in \eqref{eq:eliminate-Ik} and simplifying using the trigonometric identity for the sine of a difference, we obtain
\begin{equation} \label{eq:omkp1-zero-dyn}
    \Sm \omega_{k+1} - \Sp \omega_k = - K \sin\theta_k \left[ \frac{1}{\omega_k} + \frac{1}{\omega_{k+1}} \right]
\end{equation}

\noindent where
\begin{align}
    K &\triangleq \frac{g \delths^2}{4 R \sin{(\delths/2)}},  \quad K \geq 0 \,\, \forall \delths \in (0, \pi) \notag \\ 
    \Sm &\triangleq \sin\left(\phi - \frac{\delths}{2}\right), \, \Sp \triangleq \sin\left(\phi + \frac{\delths}{2}\right)
\end{align}

\noindent From \eqref{eq:hybrid-theta-delta} and \eqref{eq:omkp1-zero-dyn}, the zero dynamics of the system is given by the two-dimensional nonlinear map:
\begin{align} 
&\mathcal{Z}(\theta_k, \omega_k, \theta_{k+1}, \omega_{k+1}) = 0 \notag \\
&\mathcal{Z} \triangleq \begin{bmatrix} \theta_{k+1} - \theta_k - \delths \\
\Sm \omega_{k+1} - \Sp \omega_k + K \sin\theta_k \left[ \dfrac{1}{\omega_k} + \dfrac{1}{\omega_{k+1}} \right]\end{bmatrix} \label{eq:zerodyn}
\end{align}

\noindent which describes the behavior of the system when \eqref{eq:h-v-zero-dyn} is satisfied. We term the above equation the \emph{discrete zero dynamics} (DZD). An explicit solution for $\omega_{k+1}$ in terms of $\theta_k$ and $\omega_k$ is obtained in Appendix \ref{appendix:zero-dyn-sol}.

\begin{remark} \label{rem:solution-set}
    The zero dynamics induced by VHCs completely describe the evolution of the system trajectories given initial conditions that lie on the constraint manifold \cite{maggiore_virtual_2013, kant_orbital_2020}.
    Similarly, for DVHCs, the DZD defines a solution set which completely captures the evolution of system trajectories given initial conditions that satisfy the DVHC in \eqref{eq:VHC-rho} and the corresponding constraint on velocities in \eqref{eq:VHC-Drho}.
\end{remark}

\begin{remark}
    The discrete zero dynamics (DZD) presented here is different from the hybrid zero dynamics (HZD) \cite{westervelt_hybrid_2003, grizzle_virtual_2019, kao-vukovich_synthesis_2023}, which has been studied for bipedal robots subject to a VHC during their swing phase. The HZD is the combination of a flow describing the dynamics of the passive coordinate on the constraint manifold, and an \emph{impact map} describing the jump in the passive coordinate and velocity when the system trajectory intersects the \emph{switching surface}.
    In contrast, the DZD is solely a map which completely describes the evolution of the passive coordinate and velocity for a hybrid system subject to a DVHC.
\end{remark}

We investigate the DZD \eqref{eq:zerodyn} in the limit $\delths \to 0$. From \eqref{eq:time-of-flight}, $\delths \to 0 \Rightarrow \delta_k \to 0$. It follows from \eqref{eq:hybrid-theta-delta} and \eqref{eq:time-of-flight} that
\begin{equation} \label{eq:lim-theta}
\frac{\theta_{k+1} - \theta_k}{\delta_k} = \omega_{k+1} \quad
\Rightarrow \qquad \lim_{\delta_k \to 0} \frac{\theta_{k+1} - \theta_k}{\delta_k} = \dot \theta = \omega 
\end{equation}

\noindent Now, observe that \eqref{eq:omkp1-zero-dyn} can be rewritten as
\begin{equation} 
\begin{split}
       &[\omega_{k+1} - \omega_k] \sin\phi \cos (\delths/2) \\
     - &[\omega_{k+1} + \omega_k] \cos\phi \sin (\delths/2) \\
     &=  - \frac{g \delths^2}{4 R \sin{(\delths/2)}}\sin\theta_k \left[ \frac{1}{\omega_k} + \frac{1}{\omega_{k+1}} \right]
\end{split}
\end{equation}

\noindent Dividing both sides by $\delta_k$, which from \eqref{eq:time-of-flight} can alternatively be expressed as $\delths/\omega_{k+1}$, and simplifying we obtain
\begin{equation} \label{eq:omkp1-zero-dyn-rewrite}
\begin{split}
       &\frac{\omega_{k+1} - \omega_k}{\delta_k} \sin\phi \cos (\delths/2) \\
     - &\frac{\omega_{k+1}[\omega_{k+1} + \omega_k]}{2} \cos\phi \frac{\sin (\delths/2)}{(\delths/2)} \\
     &= - \frac{g (\delths/2)}{R \sin{(\delths/2)}}\sin\theta_k \times \frac{1}{2}\left[ \frac{\omega_{k+1}}{\omega_k} + 1 \right]
\end{split}
\end{equation}

\noindent Observe that 
\begin{equation}
    \lim_{\delta_k \to 0} \frac{\omega_{k+1} - \omega_k}{\delta_k} = \dot\omega, \quad 
    \lim_{\delths \to 0} \frac{\sin (\delths/2)}{(\delths/2)} = 1
\end{equation}

\noindent and in the limit $\delths \to 0$, $\omega_{k+1} \approx \omega_k$, \emph{i.e.}
\begin{equation*}
    \frac{\omega_{k+1}[\omega_{k+1} + \omega_k]}{2} \approx \omega^2, \quad
    \frac{1}{2}\left[ \frac{\omega_{k+1}}{\omega_k} + 1 \right] \approx 1
\end{equation*}

\noindent Thus, in the limit $\delths \to 0$, \eqref{eq:omkp1-zero-dyn-rewrite} becomes
\begin{equation}
    \dot \omega \sin\phi - \omega^2 \cos\phi = - \frac{g}{R} \sin\theta
\end{equation}

\noindent Using \eqref{eq:lim-theta}, the above equation can be rewritten as
\begin{equation} \label{eq:lim-zero-dyn}
    \ddot \theta = -\frac{g \sin\theta}{R \sin\phi} + \cot\phi \dot\theta^2
\end{equation}

\noindent which is equivalent to the zero dynamics \cite{khandelwal_propeller_2025} for continuous-time propeller motion using normal forcing. 

If $\phi \in \{-\pi/2, \pi/2\}$, \eqref{eq:omkp1-zero-dyn} can be simplified to the form
\begin{equation} \label{eq:omkp1-zero-dyn-piby2}
    \omega_{k+1} - \omega_k = P \sin\theta_k \left[ \frac{1}{\omega_k} + \frac{1}{\omega_{k+1}} \right]
\end{equation}

\noindent where
\begin{equation}
    P \triangleq \begin{cases}
         &\,\,\,\,\, \dfrac{g \delths^2}{2 R \sin{(\delths)}}, \quad \phi = -\pi/2 \\
        &- \dfrac{g \delths^2}{2 R \sin{(\delths)}}, \quad \phi = \pi/2
    \end{cases}
\end{equation}

\noindent Therefore, it follows from \eqref{eq:hybrid-theta-delta} and \eqref{eq:omkp1-zero-dyn-piby2} that the DZD in \eqref{eq:zerodyn} takes the form
\begin{align}
&\mathcal{Z}(\theta_k, \omega_k, \theta_{k+1}, \omega_{k+1}) = 0 \notag \\
&\mathcal{Z} \triangleq \begin{bmatrix} \theta_{k+1} - \theta_k - \delths \\
\omega_{k+1} - \omega_k - P \sin\theta_k \left[ \dfrac{1}{\omega_k} + \dfrac{1}{\omega_{k+1}} \right]\end{bmatrix}  \label{eq:zerodyn-special}
\end{align}

It follows from \eqref{eq:lim-zero-dyn} that the DZD \eqref{eq:zerodyn-special} in the limit $\delths \to 0$ takes the form \cite{khandelwal_propeller_2025}
\begin{equation} \label{eq:lim-zero-dyn-piby2}
    \ddot \theta = \begin{cases}
        &\,\,\,\,\, \dfrac{g \sin \theta}{R}, \quad \phi = -\pi/2 \\[2ex]
        &- \dfrac{g \sin \theta}{R}, \quad \phi =  \pi/2
    \end{cases}
\end{equation}

\begin{remark}
    The behavior of the zero dynamics \eqref{eq:lim-zero-dyn} for continuous-time propeller motion is well understood \cite{khandelwal_propeller_2025}; periodic solutions are obtained only when $\phi \in \{-\pi/2, \pi/2\}$, \emph{i.e.}, when the zero dynamics is described by \eqref{eq:lim-zero-dyn-piby2}. 
    Similarly, the DZD \eqref{eq:zerodyn} permits periodic solutions only when $\phi \in \{-\pi/2, \pi/2\}$, \emph{i.e.}, when it reduces to \eqref{eq:zerodyn-special}.
    An in-depth analytical and numerical investigation of the DZD is provided in Appendix \ref{appendix:zero-dyn}, which also discusses additional conditions for periodic behavior.
\end{remark}

In particular, periodic solutions to \eqref{eq:zerodyn-special} are obtained only when $\delths = (p/q) 2\pi$, $p, q \in Z^+$, and $\theta_k \!\!\mod 2\pi \in \{0, \delths/2\}$ for some $k$. This is proved for the special case $\delths = 2\pi/N$, $N \in Z^+$ - see Theorem \ref{thm:periodic} in Appendix \ref{sec44}.
We now describe a procedure to stabilize a desired orbit when $\phi = \pm \pi/2$.

\section{Stabilization of Propeller Motion} \label{sec5}

\subsection{Orbit Selection}

A distinct discrete propeller motion described by the DVHCs in \eqref{eq:VHC-Phi} and corresponding constraints on velocity in \eqref{eq:VHC-Psi}, with $\phi \in \{-\pi/2, \pi/2\}$ and $\delths = 2\pi/N$ (see \eqref{eq:delths} in Appendix \ref{sec44}) is the orbit
\begin{equation} \label{eq:orbit}
\begin{split}
    \mathcal{O}^* = \{(q, \dot q) : \qac(k) &= \Phi(\theta_k), \dqac(k) = \Psi(\theta_k, \omega_k), \\
    \omega_k &= \omega^* \,\,\forall\, \theta_k \!\!\!\!\mod 2\pi = \theta^* \}
\end{split}
\end{equation}

\noindent with $\omega^*$ chosen such that Assumption \ref{asm:omk-min-periodic} (in Appendix \ref{sec44}) is satisfied. To realize stabilization of $\mathcal{O}^*$, we define the section
\begin{equation} \label{eq:poincare-section}
    \Sigma = \{(q, \dot q) \in \mathbb{R}^6 : \theta \!\!\!\! \mod 2\pi = \theta^*, \omega > 0\}
\end{equation}

\noindent on which the states are
\begin{equation}
    z = \begin{bmatrix} \qac^T & \dqac^T &\omega \end{bmatrix}^T, \quad z \in \mathbb{R}^5
\end{equation}

\noindent For a system trajectory not on $\mathcal{O}^*$, the inputs $I$ and $r$ applied when $\theta_k \!\!\mod 2\pi = \theta^*$ are modified from their values in \eqref{eq:Ik} and \eqref{eq:rk} to drive the trajectory to $\mathcal{O}^*$. Since $\theta_{k + N} \!\!\mod 2\pi = \theta_k \!\!\mod 2\pi$, the hybrid dynamics between successive intersections of the system trajectory with $\Sigma$, which involves $N$ applications of the impulsive force, may be expressed by
\begin{equation} \label{eq:hybrid-map}
    z(j+1) = \mathbb{P} [z(j), I(j), r(j)]
\end{equation}

\noindent where $z(j)$ denotes the states on $\Sigma$ immediately prior to application of the control inputs $I(j)$ and $r(j)$ on $\Sigma$. It should be noted that for every intersection of the system trajectory with $\Sigma$, the increment in $k$ is $N$ whereas the increment in $j$ is $1$.

\subsection{Orbital Stabilization}

If $(q, \dot q) \in \mathcal{O}^*$, the system trajectory evolves on $\mathcal{O}^*$ under the action of the control inputs from \eqref{eq:Ik} and \eqref{eq:rk}. Therefore, the intersection of $\mathcal{O}^*$ with $\Sigma$ is a fixed point $z(j) = z^*$, $I(j) = I^*$, $r(j) = r^*$ of the map $\mathbb{P}$, 
where $I^*$ and $r^*$ are obtained from \eqref{eq:Ik} and \eqref{eq:rk} respectively with $\theta_k = \theta^*$, $\omega_k = \omega^*$, $\rho_k = 0$, and $\drho_k = 0$:
\begin{equation} \label{eq:fixed-point}
    z^* = \mathbb{P} (z^*, I^*, r^*)
\end{equation}

\noindent If $(q, \dot q) \notin \mathcal{O}^*$, the inputs from \eqref{eq:Ik} and \eqref{eq:rk} ensure that $\rho_k \to 0$, but system trajectories do not necessarily converge to $\mathcal{O}^*$. We apply the ICPM approach to stabilize the fixed point $z^*$ on $\Sigma$, and consequently the orbit $\mathcal{O}^*$; we linearize $\mathbb{P}$ about $z(j) = z^*$ and $I(j) = I^*,\, r(j) = r^*$:
\begin{align} 
    e(j+1) &= \mathcal{A} e(j) + \mathcal{B} u(j) \label{eq:linearized-map} \\ 
    e(j) &\triangleq z(j) - z^*, \quad 
    u(j) \triangleq \begin{bmatrix} I(j) \\ r(j) \end{bmatrix} - \begin{bmatrix} I^* \\ r^* \end{bmatrix}
\end{align} 

\noindent where $\mathcal{A} \in \mathbb{R}^{5 \times 5}$, $\mathcal{B} \in \mathbb{R}^{5 \times 2}$ are computed numerically.
The $i$-th column $\mathcal{A}_i$ of $\mathcal{A}$ is computed as:
\begin{equation}
    \mathcal{A}_i = \frac{1}{\epsilon_1} \left[ \mathbb{P} (z^* + \gamma_i, I^*, r^*) - z^* \right]
\end{equation}

\noindent where $\gamma_i$ is the $i$-th column of $\epsilon_1 \mathbb{I}_5$; $\epsilon_1$ is a small number, and $\mathbb{I}_5$ is the $5 \times 5$ identity matrix. 
The first and second columns of $\mathcal{B}$ are computed as:
\begin{subequations}
\begin{align}
    \mathcal{B}_1 = \frac{1}{\epsilon_2} \left[ \mathbb{P} (z^*, I^* + \epsilon_2, r^*) - z^* \right] \\
    \mathcal{B}_2 = \frac{1}{\epsilon_2} \left[ \mathbb{P} (z^*, I^*, r^* + \epsilon_2) - z^* \right]
\end{align}
\end{subequations}

\noindent where $\epsilon_2$ is a small number. If $(\mathcal{A}, \mathcal{B})$ is controllable, the orbit $\mathcal{O}^*$ can be asymptotically stabilized by the discrete feedback
\begin{equation} \label{eq:discrete-feedback}
    u(j) = \mathcal{K} e(j)
\end{equation}

\noindent where $\mathcal{K}$ is chosen to place the eigenvalues of $(\mathcal{A} + \mathcal{B}\mathcal{K})$ inside the unit circle. Thus, for every intersection of the system trajectory with $\Sigma$, the applied inputs are given by
\begin{equation} \label{eq:orbit-I-r}
    \begin{bmatrix} I(j) \\ r(j) \end{bmatrix} = \begin{bmatrix} I_k \\ r_k \end{bmatrix} + u(j)
\end{equation}

\noindent The additional input $u(j)$ is inactive when the system trajectory is sufficiently close to $\mathcal{O}^*$, \emph{i.e.}, $\norm{e(j)}$ is sufficiently small.

\begin{remark}
    For propeller motion in continuous-time, it was shown that orbital stabilization could be realized by varying only the applied force without a discontinuous change in its point of application when the system trajectory intersected a chosen Poincar\'e section \cite{khandelwal_propeller_2025, kant_orbital_2020}. However, numerical results show that both the applied impulse and its point of application need to be varied on $\Sigma$ to stabilize an orbit for propeller motion using purely impulsive inputs as the pair $(\mathcal{A}, \mathcal{B}_1)$ is found not to be controllable. 

    Further, numerical results show that the pair $(\mathcal{A}, \mathcal{B})$ is not controllable when $\lambda_x = \lambda_y = 0$. This indicates that orbital stabilization is not possible when a deadbeat controller is used to enforce the DVHC.
\end{remark}

\section{Simulation} \label{sec6}

The physical parameters of the devil-stick in SI units are chosen to be:
\begin{equation*}
    m = 0.1, \quad \ell = 0.5, \quad J = \frac{1}{12} m \ell^2 = 0.0021
\end{equation*}

\noindent For all simulations in this section, $\delths$ is given by \eqref{eq:delths} in Appendix \ref{sec44} with the choice $N = 6$.

\subsection{Enforcing the DVHC}

We consider the DVHC given by \eqref{eq:VHC-Phi} with 
\begin{equation} \label{eq:vhc-sim}
    R = 1, \quad \phi = \pi/2
\end{equation}

\noindent It is known from Theorem \ref{thm:periodic} that the resulting DZD given by \eqref{eq:zerodyn-special} permits stable periodic orbits if $\theta_1$ is chosen to ensure that $\theta_k \!\!\mod 2\pi \in \{0, \delths/2\}$ for some $k$. 
We simulate the behavior of the system under the control inputs \eqref{eq:Ik} and \eqref{eq:rk}, with the values 
\begin{equation} \label{eq:lambda-sim}
    \lambda_x = \lambda_y = 0.5
\end{equation}

\begin{figure}[t]
    \centering
    \psfrag{A}{\hspace{-10pt} \footnotesize{$\theta$ (rad)}}
    \psfrag{B}{\hspace{-18pt} \footnotesize{$\omega$ (rad/s)}}
    \psfrag{T}{\footnotesize{$k$}}
    \psfrag{Q}{\footnotesize{$\bar E_k$}}
    \psfrag{E}{\footnotesize{$\bar E = $}}
    \psfrag{F}{\hspace{-10pt} \footnotesize{$I_k$ (Ns)}}
    \psfrag{R}{\hspace{-10pt} \footnotesize{$r_k$ (m)}}
    \psfrag{C}{$\rho_x$}
    \psfrag{D}{$\rho_y$}
    \psfrag{G}{$\drho_x$}
    \psfrag{L}{$\drho_y$}
    \psfrag{O}{$\mathcal{O}^*$}
    \includegraphics[width=\linewidth]{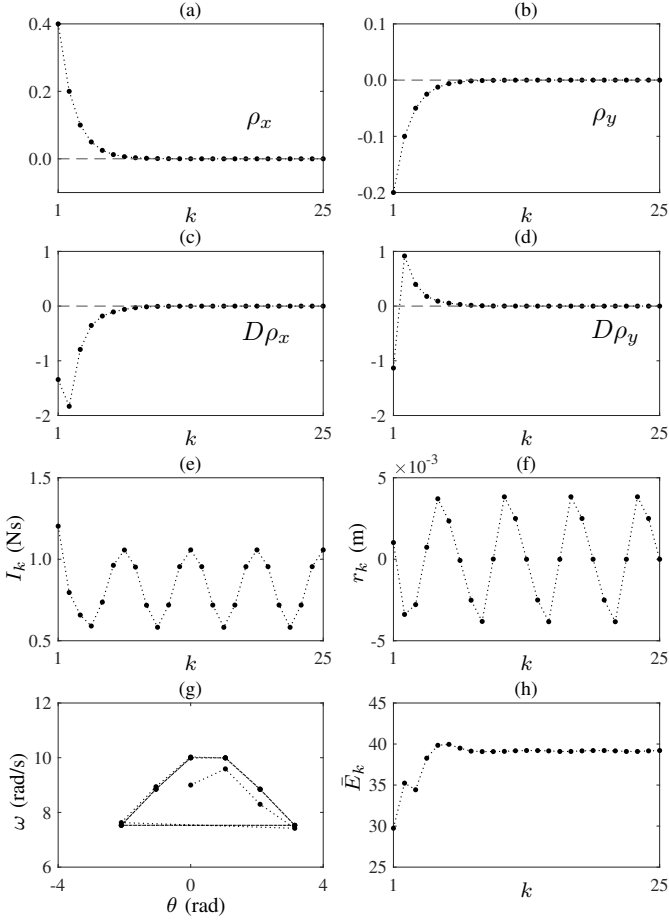}
    \caption{Stabilization of propeller motion of a devil-stick from initial conditions \eqref{eq:initial-conditions}: 
    (a)-(b) show the components of $\rho_k$, (c)-(d) show the components of $\drho_k$,
    (e) shows the applied impulse $I_k$, (f) shows the point of application $r_k$ of the impulsive force,
    (g) shows the evolution of the system trajectory $(\theta_k, \omega_k)$ in the $\theta$-$\omega$ plane with $\theta$ plotted in the interval $(-\pi, \pi]$, and (h) shows the value of $\bar E_k$ from \eqref{eq:invariant-piby2}.}
    \label{fig:sim-vhc}
\end{figure}

\begin{figure}[t]
    \centering
    \psfrag{A}{\hspace{-10pt} \footnotesize{$\theta$ (rad)}}
    \psfrag{B}{\hspace{-18pt} \footnotesize{$\omega$ (rad/s)}}
    \psfrag{T}{\footnotesize{$k$}}
    \psfrag{Q}{\footnotesize{$\bar E_k$}}
    \psfrag{E}{\footnotesize{$\bar E = $}}
    \psfrag{F}{\hspace{-10pt} \footnotesize{$I_k$ (Ns)}}
    \psfrag{R}{\hspace{-10pt} \footnotesize{$r_k$ (m)}}
    \psfrag{C}{$\rho_x$}
    \psfrag{D}{$\rho_y$}
    \psfrag{G}{$\drho_x$}
    \psfrag{L}{$\drho_y$}
    \psfrag{O}{$\mathcal{O}^*$}
    \includegraphics[width=\linewidth]{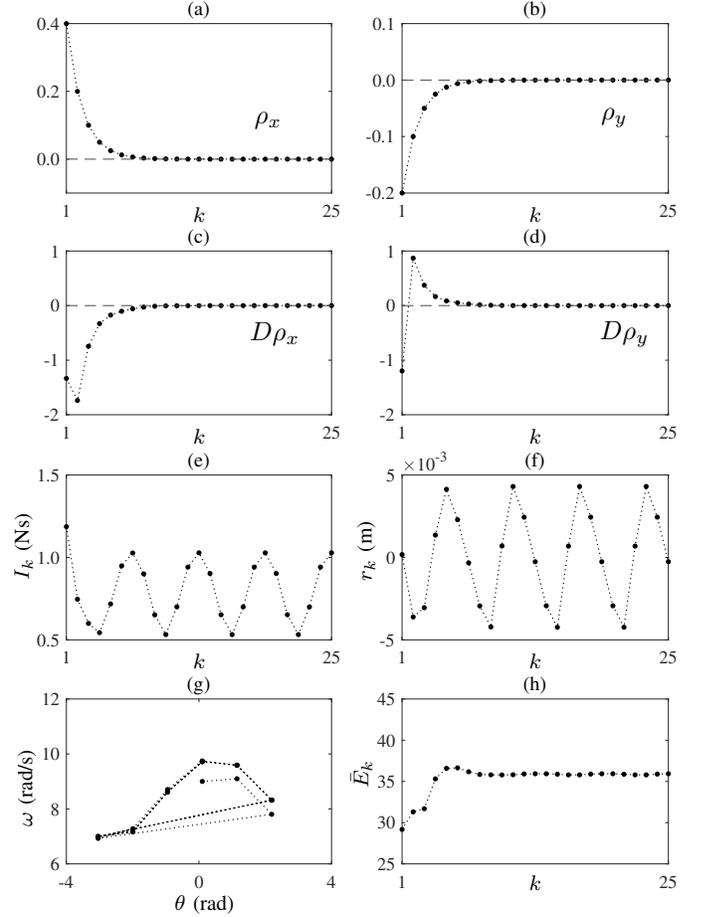}
    \caption{Stabilization of propeller motion of a devil-stick from initial conditions \eqref{eq:initial-conditions-np}: 
    (a)-(b) show the components of $\rho_k$, (c)-(d) show the components of $\drho_k$,
    (e) shows the applied impulse $I_k$, (f) shows the point of application $r_k$ of the impulsive force,
    (g) shows the evolution of the system trajectory $(\theta_k, \omega_k)$ in the $\theta$-$\omega$ plane with $\theta$ plotted in the interval $(-\pi, \pi]$, and (h) shows the value of $\bar E_k$ from \eqref{eq:invariant-piby2}.}
    \label{fig:sim-vhc-np}
\end{figure}

We first consider the initial conditions with $\theta_1 = 0$:
\begin{equation} \label{eq:initial-conditions}
    \begin{bmatrix} q^T & \dot q^T \end{bmatrix}^T = \begin{bmatrix} 0.4 & -1.2  & 0.0 & 6.1 & -6.0 & 9.0 \end{bmatrix}^T 
\end{equation}

\noindent for which $\rho_1 \neq 0$. The results of the simulation are shown in Fig.\ref{fig:sim-vhc} for 4 complete rotations of the devil-stick, \emph{i.e.} $k = 1$ through $k = 25$, a duration of approx. $2.91$ s. The components of $\rho_k$ are shown in Fig.\ref{fig:sim-vhc}(a)-(b), and demonstrate exponential convergence of the system trajectory to $\rho_k = 0$. The components of $\drho_k$ are shown in Fig.\ref{fig:sim-vhc}(c)-(d), and it can be seen that $\drho_k \to 0$ as $\rho_k \to 0$. As the DVHC is imposed, it is seen that the system trajectory converges to the periodic orbit given by
\begin{equation}
\begin{split}
    \mathcal{O} = \{(q, \dot q) : \qac(k) &= \Phi(\theta_k), \dqac(k) = \Psi(\theta_k, \omega_k), \\
    \omega_k &= 9.9963 \,\,\forall\, \theta_k \!\!\!\!\mod 2\pi = 0 \}
\end{split} 
\end{equation}

\noindent The applied impulse $I_k$ and its point of application $r_k$ are shown in Fig.\ref{fig:sim-vhc}(e) and Fig.\ref{fig:sim-vhc}(f) respectively, and they take on periodic values as the system trajectory converges to the orbit $\mathcal{O}$ above. The evolution of the system trajectory in the $\theta$-$\omega$ plane is shown in Fig.\ref{fig:sim-vhc}(g). Finally, Fig.\ref{fig:sim-vhc}(h) shows the plot of $\bar E_k$ from \eqref{eq:invariant-piby2} in Appendix \ref{appendix-invariant}; it is seen that the value of $\bar E_k$ becomes approximately invariant as the trajectory converges to $\mathcal{O}$.

We also simulate the behavior of the system from the initial conditions
\begin{equation} \label{eq:initial-conditions-np}
    \begin{bmatrix} q^T & \dot q^T \end{bmatrix}^T = \begin{bmatrix} 0.4 & -1.2  &       0.1 & 6.1 & -6.0 & 9.0 \end{bmatrix}^T 
\end{equation}

\noindent Since $\theta_k \!\!\mod 2\pi \notin \{0, \delths/2\}$ for any $k$ with this choice of $\theta_1$, the reduced system does not permit a stable, periodic orbit. However, the DVHC can still be enforced, as demonstrated by Fig.\ref{fig:sim-vhc-np}, which shows the results for 4 rotations of the devil-stick, \emph{i.e.} $k = 1$ through $k = 25$, a duration of approx. $3.05$ s. The components of $\rho_k$ in Fig.\ref{fig:sim-vhc-np}(a)-(b), and of $\drho_k$ in Fig.\ref{fig:sim-vhc-np}(c)-(d) again converge to $0$. The values of $I_k$ and $r_k$ are shown in  Fig.\ref{fig:sim-vhc-np}(e) and  Fig.\ref{fig:sim-vhc-np}(f), respectively. The evolution of the system trajectory in the $\theta$-$\omega$ plane is shown in Fig.\ref{fig:sim-vhc-np}(g). Finally, Fig.\ref{fig:sim-vhc-np}(h) shows the plot of $\bar E_k$ from \eqref{eq:invariant-piby2} in Appendix \ref{appendix-invariant}. It can be verified that the system trajectory does not converge to a periodic orbit, but exhibits a gradual drift in $\omega$ as the devil-stick rotates.

\begin{remark}
    Once the DVHC is stabilized, $\omega_k$ and $\bar E_k$ in Fig.\ref{fig:sim-vhc} assume periodic behavior, whereas $\omega_k$ and $\bar E_k$ in Fig.\ref{fig:sim-vhc-np} drift gradually over a large number of simulation steps. This is an artifact of the choice of $\theta_1$ for the two simulations, and agrees with the results in Appendices \ref{sec44} and \ref{sec45}.
    The choice of $\delths$ and the initial conditions determine the magnitude and direction of the drift in $\omega_k$ and $\bar E_k$. If $\omega_k$ (consequently, $\bar E_k$) decreases gradually over time, a point is eventually reached when real solutions to \eqref{eq:omkp1-quadratic-soln-piby2} cannot be found. 
    If $\omega_k$ (consequently, $\bar E_k$) increases gradually over time, the rate of increase decreases over time as can be seen in Fig.\ref{fig:zero-rational} in Appendix \ref{sec45}. 
\end{remark}

\subsection{Stabilization of a Periodic Orbit}

\begin{figure}[t]
    \centering
    \psfrag{A}{\hspace{-10pt} \footnotesize{$\theta$ (rad)}}
    \psfrag{B}{\hspace{-18pt} \footnotesize{$\omega$ (rad/s)}}
    \psfrag{T}{\footnotesize{$k$}}
    \psfrag{Q}{\footnotesize{$\bar E_k$}}
    \psfrag{E}{\footnotesize{$\bar E = $}}
    \psfrag{F}{\hspace{-10pt} \footnotesize{$I_k$ (Ns)}}
    \psfrag{R}{\hspace{-10pt} \footnotesize{$r_k$ (m)}}
    \psfrag{C}{$\rho_x$}
    \psfrag{D}{$\rho_y$}
    \psfrag{G}{$\drho_x$}
    \psfrag{L}{$\drho_y$}
    \psfrag{O}{$\mathcal{O}^*$}
    \includegraphics[width=\linewidth]{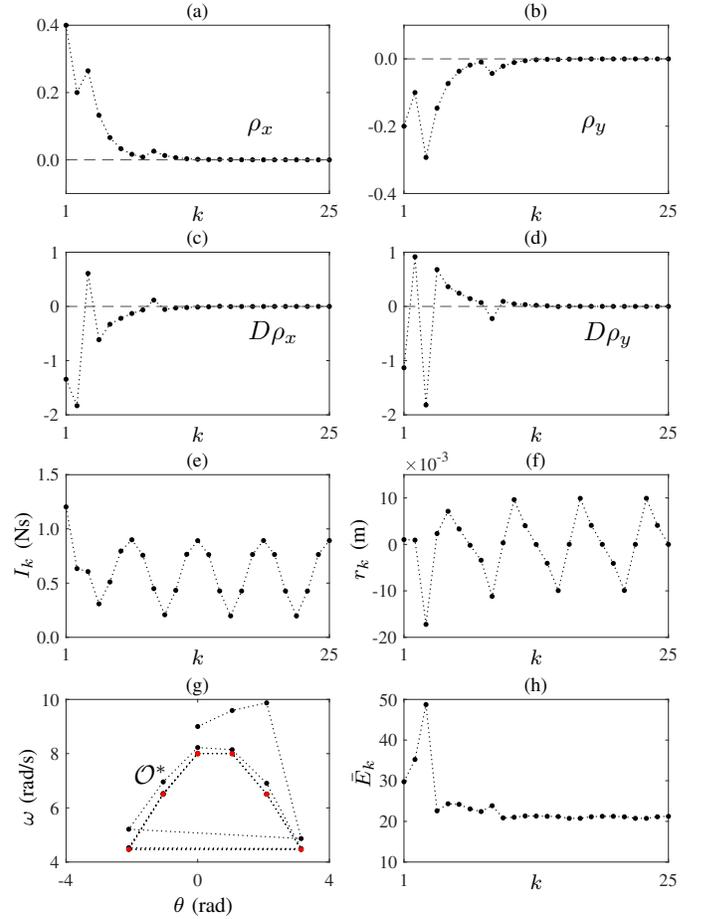}
    \caption{Orbital stabilization of propeller motion of a devil-stick from arbitrary initial conditions: 
    (a)-(b) show the components of $\rho_k$, (c)-(d) show the components of $\drho_k$,
    (e) shows the applied impulse $I_k$, (f) shows the point of application $r_k$ of the impulsive force,
    (g) shows the evolution of the system trajectory $(\theta_k, \omega_k)$ in the $\theta$-$\omega$ plane with $\theta$ plotted in the interval $(-\pi, \pi]$, and (h) shows the value of $\bar E_k$ from \eqref{eq:invariant-piby2}.}
    \label{fig:sim-orbit}
\end{figure}

Having demonstrated the effectiveness of the controller from Section \ref{sec33} in enforcing the DVHCs, we demonstrate the efficacy of the control design in Section \ref{sec5} in stabilizing a desired orbit. We again consider the DVHC \eqref{eq:VHC-Phi} with $\delths = 2\pi/6$ and the parameters in \eqref{eq:vhc-sim}.
The desired orbit to be stabilized is the one passing through
\begin{equation}
    \begin{bmatrix} q^T & \dot q^T \end{bmatrix}^T = \begin{bmatrix} 0 & -1 & 0 & 6.6159 & -4.4168 & 8 \end{bmatrix}^T
\end{equation}

\noindent which is the orbit 
\begin{equation} \label{eq:orbit-sim}
\begin{split}
    \mathcal{O}^* = \{(q, \dot q) : \qac(k) &= \Phi(\theta_k), \dqac(k) = \Psi(\theta_k, \omega_k), \\
    \omega_k &= 8 \,\,\forall\, \theta_k \!\!\!\!\mod 2\pi = \pi/3 \}
\end{split}
\end{equation}

\noindent The control inputs \eqref{eq:Ik} and \eqref{eq:rk} are obtained with the choice of $\lambda$ in \eqref{eq:lambda-sim}. For the stable periodic orbit \eqref{eq:orbit-sim}, it can be verified that feasible solutions to \eqref{eq:delta-k} (consequently, \eqref{eq:Ik} and \eqref{eq:rk}) can be found everywhere along the orbit. The same is true for solutions to \eqref{eq:omkp1-zero-dyn-piby2} everywhere along the orbit. The section on which the control input \eqref{eq:orbit-I-r} is applied is chosen to be
\begin{equation} \label{eq:poincare-section-sim}
    \Sigma = \{(q, \dot q) \in \mathbb{R}^6 : \theta \!\!\!\!\mod 2\pi = \pi/3, \omega > 0\}
\end{equation}

\noindent Corresponding to the orbit in \eqref{eq:orbit-sim} and the section \eqref{eq:poincare-section-sim}, the fixed point of the map $\mathbb{P}$ from \eqref{eq:fixed-point} is
\begin{align}
    z^* &= \begin{bmatrix}
        0.8660 & -0.5000 & 6.6159 & 3.1777 & 8
    \end{bmatrix}^T \notag \\
    I^* &= 0.7639, \quad r^* = -0.0041
\end{align}

\noindent The values of $\mathcal{A}$ and $\mathcal{B}$ in \eqref{eq:linearized-map} were obtained numerically as
\begin{equation*}
\begin{split}
    \mathcal{A} &= \begin{bmatrix}
    0.0156  &  0.0000  &  0.0000  &  0.0000  &  0.0000 \\
   -0.0000  &  0.0156  & -0.0000  & -0.0000  & -0.0000 \\
    2.1827  &  1.4297  &  0.4612  &  0.7989  &  0.0000 \\ 
    1.5525  &  0.8448  &  0.3111  &  0.5388  & -0.0000 \\
    2.7837  &  1.7288  &  0.5577  &  0.9660  & -0.0000
    \end{bmatrix}
    \\ 
    \mathcal{B} &= \begin{bmatrix}
    -0.0436  &  0.0334  &  8.7712  &  5.4318  & 10.2021 \\
     0.0000  & -1.4029  &-27.6739  & -8.0578  &-33.4632
    \end{bmatrix}^T
\end{split}
\end{equation*}

\noindent The eigenvalues of $\mathcal{A}$ do not all lie within the unit circle. It can be verified that the pair $(\mathcal{A}, \mathcal{B})$ is controllable, and the gain matrix $\mathcal{K}$ in \eqref{eq:discrete-feedback} that asymptotically stabilizes the closed-loop system is:
\begin{equation*}
    \mathcal{K} = \begin{bmatrix}
    -0.2439  & -0.1185 &  -0.0486  & -0.0842  &  0.0000 \\
     0.0066  &  0.0153  &  0.0018  &  0.0030  &  0.0000
    \end{bmatrix}
\end{equation*}

\noindent obtained using LQR with the weighting matrices
\begin{equation*}
    \mathcal{Q} = \mathbb{I}_5, \quad \mathcal{R} = 2 \mathbb{I}_2
\end{equation*}

The same initial conditions \eqref{eq:initial-conditions} are used, and the simulation results are shown in Fig.\ref{fig:sim-orbit} for 4 rotations of the devil stick, \emph{i.e.} $k = 1$ through $k = 25$, a duration of approx. $4.05$ s. The components of $\rho_k$ are shown in Fig.\ref{fig:sim-orbit}(a)-(b), and demonstrate convergence of the system trajectory to $\rho_k = 0$. Unlike in Fig.\ref{fig:sim-vhc}(a)-(b), the convergence is not uniform owing to the additional control inputs being applied at $\Sigma$. The components of $\drho_k$ are shown in Fig.\ref{fig:sim-orbit}(c)-(d); again $\drho_k \to 0$ as $\rho_k \to 0$. 
The applied impulse $I_k$ and its point of application $r_k$ are shown in Fig.\ref{fig:sim-orbit}(e) and Fig.\ref{fig:sim-orbit}(f) respectively; note that these inputs are obtained from \eqref{eq:orbit-I-r} for $j = 1, 2, 3$, which corresponds to $k = 2, 8, 14$. For $j > 3$, the controller $u(j)$ is inactive as the system trajectory is sufficiently close to $\mathcal{O}^*$. The control inputs take on periodic values as the system trajectory converges to the desired orbit $\mathcal{O}^*$
The evolution of the system trajectory in the $\theta$-$\omega$ plane is shown in Fig.\ref{fig:sim-orbit}(g), with the periodic points associated with $\mathcal{O}^*$ shown in red. Finally, Fig.\ref{fig:sim-orbit}(h) shows the plot of $\bar E_k$ from \eqref{eq:invariant-piby2}; it is seen that the value of $\bar E_k$ becomes approximately invariant as the trajectory converges to $\mathcal{O}^*$. 
The trajectory of the center-of-mass of the devil-stick is shown in Fig.\ref{fig:trajectory}. As the orbit $\mathcal{O}^*$ is stabilized, the center-of-mass is seen to lie on the curve defined by the DVHC in \eqref{eq:VHC-Phi} only at the instants impulsive inputs are applied.

\begin{figure}[t]
    \centering
    \psfrag{X}{\hspace{-10pt} \footnotesize{$h_x$ (m)}}
    \psfrag{Y}{\hspace{-10pt} \footnotesize{$h_y$ (m)}}
    \includegraphics[width=0.72\linewidth]{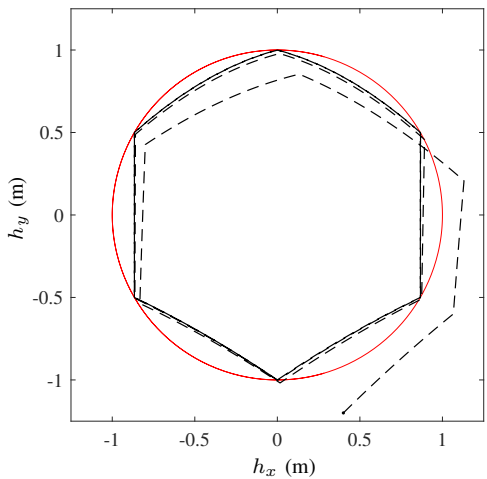}
    \caption{Trajectory of the center-of-mass of the devil-stick corresponding to the results in Fig.\ref{fig:sim-orbit}, with the DVHC in \eqref{eq:VHC-Phi} shown in red.}
    \label{fig:trajectory}
\end{figure}

\subsection{Aperiodic Motion}

\begin{figure}[t]
    \centering
    \psfrag{A}{\hspace{-10pt} \footnotesize{$\theta$ (rad)}}
    \psfrag{B}{\hspace{-18pt} \footnotesize{$\omega$ (rad/s)}}
    \psfrag{T}{\footnotesize{$k$}}
    \psfrag{Q}{\footnotesize{$\bar E_k$}}
    \psfrag{E}{\footnotesize{$\bar E = $}}
    \psfrag{F}{\hspace{-10pt} \footnotesize{$I_k$ (Ns)}}
    \psfrag{R}{\hspace{-10pt} \footnotesize{$r_k$ (m)}}
    \psfrag{C}{$\rho_x$}
    \psfrag{D}{$\rho_y$}
    \psfrag{G}{$\drho_x$}
    \psfrag{L}{$\drho_y$}
    \psfrag{O}{$\mathcal{O}^*$}
    \includegraphics[width=\linewidth]{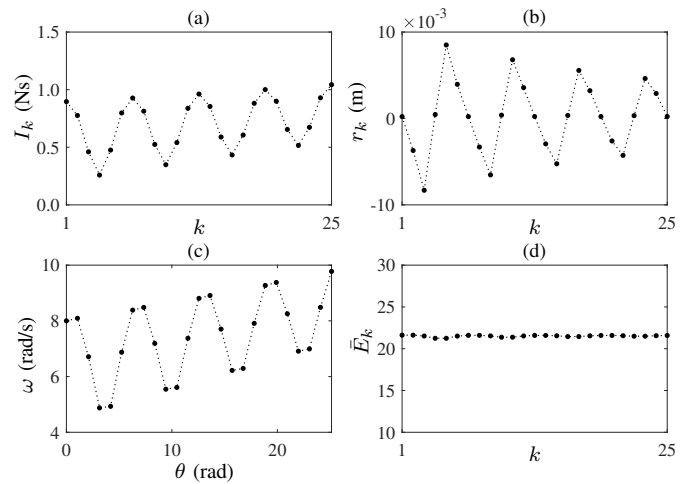}
    \caption{Aperiodic propeller motion of a devil-stick: 
    (a) shows the applied impulse $I_k$, (b) shows the point of application $r_k$ of the impulsive force,
    (c) shows the evolution of the system trajectory $(\theta_k, \omega_k)$ in the $\theta$-$\omega$ plane, and (d) shows the value of $\bar E_k$ from \eqref{eq:invariant}.}
    \label{fig:sim-aperiodic}
\end{figure}

We present a simulation with the DVHC \eqref{eq:VHC-Phi}, $\delths = 2\pi/6$, and parameters
\begin{equation}
    R = 1, \quad \phi = \pi/2 - 0.01
\end{equation}

\noindent Since $\phi \notin \{-\pi/2, \pi/2\}$, periodic orbits do not exist. We simulate the behavior of the system from the initial condition 
\begin{equation}
    \begin{bmatrix} q^T & \dot q^T \end{bmatrix}^T = \begin{bmatrix} 0 & -1 & 0 & 6.6538 & -4.3954 & 8 \end{bmatrix}^T
\end{equation}

\noindent for which $\rho_1 = 0$ and $\drho_1 = 0$. The results of the simulation are shown in Fig.\ref{fig:sim-aperiodic} for 4 rotations of the devil-stick, a duration of approx. $3.50$ s. The components of $\rho_k$ and $\drho_k$ remain steady at zero and are therefore not plotted. The control inputs \eqref{eq:Ik} and \eqref{eq:rk} are plotted in Fig.\ref{fig:sim-aperiodic}(a) and Fig.\ref{fig:sim-aperiodic}(b) showing their aperiodic behavior. The system trajectory in the $\theta$-$\omega$ plane is shown in Fig.\ref{fig:sim-aperiodic}(c), showing the increase in $\omega$ over every rotation. The value of $\bar E_k$ from \eqref{eq:invariant} is plotted in Fig.\ref{fig:sim-aperiodic}(d); it remains approximately invariant.

\section{Conclusion}

The hybrid dynamics of juggled objects controlled using impacts can be expressed using discrete-time maps. The concept of discrete virtual holonomic constraints is introduced to design trajectories for such underactuated systems described by discrete-time maps, in the context of propeller motion of a devil-stick. It is shown that the DVHCs can be enforced using discrete feedback in the form of impulsive inputs applied when the passive coordinate of the system evolves by a specific value. With the DVHCs enforced, the dynamics of the system reduces to the discrete zero dynamics, which is a two-dimensional discrete-time autonomous system involving the passive coordinate and velocity. The stability characteristics of the full system follow from the DZD, which is extensively analyzed for propeller motion of a devil-stick. It is shown that with appropriate choices of the DVHC parameters and initial conditions, the system permits a wealth of stable orbits with different rotational speeds. A desired orbit can be stabilized using the impulse controlled Poincar\'e map approach, by varying the impulsive input when the system trajectory intersects a chosen section, which occurs once every rotation of the devil-stick. Extensive simulation results demonstrate the efficacy of the controller enforcing the DVHCs, and the controller stabilizing a specific orbit.

The approach presented in this paper is generalizable to a larger class of systems manipulated using purely impulsive inputs, when only intermittent control over the configuration of the object is required. Future work is focused on realizing this generalization for a larger class of mechanical systems.
Further analysis of the zero dynamics for this problem and the general problem is also in progress. This paper has derived a model for \emph{rotations} of the passive coordinate, and we are investigating how it may be generalized to \emph{oscillations}, in which the sign of the passive velocity changes. 

\appendices
\section{Analysis of Discrete Zero Dynamics} \label{appendix:zero-dyn} 

\subsection{Quadratic Solution} \label{appendix:zero-dyn-sol}

The terms in \eqref{eq:omkp1-zero-dyn} can be rearranged to form a quadratic equation which can be solved for $\omega_{k+1}$ as
\begin{equation} \label{eq:omkp1-quadratic-soln}
\begin{split}
    &\omega_{k+1} = \frac{\Sp \omega_k}{2 \Sm} - \frac{K \sin\theta_k}{2 \Sm \omega_k} \\
    &\ \pm \frac{\Sp \omega_k}{2 \Sm}
    \sqrt{1 -  \frac{2 K (\Sp + 2 \Sm) \sin\theta_k}{\Sp^2 \omega_k^2} + \frac{K^2 \sin^2\theta_k}{\Sp^2 \omega_k^4}}
\end{split}
\end{equation}

\noindent The solution is real when 
\begin{equation*}
    1 -  \frac{2 K (\Sp + 2 \Sm) \sin\theta_k}{\Sp^2 \omega_k^2} + \frac{K^2 \sin^2\theta_k}{\Sp^2 \omega_k^4} \geq 0 \quad \forall k
\end{equation*}

\noindent which imposes lower bounds on the values of $\omega_k$ given $\theta_k$. It must now be ascertained whether the positive or negative square root solution in \eqref{eq:omkp1-quadratic-soln} must be used. In particular, we observe that if $\theta_k \!\!\mod 2\pi \in \{0, \pi\}$, it follows from \eqref{eq:omkp1-zero-dyn} that
\begin{equation*} 
    \omega_{k+1} = \frac{\Sp \omega_k}{\Sm}
\end{equation*}

\noindent which is obtained on selecting the positive square root solution in \eqref{eq:omkp1-quadratic-soln}. Selecting the negative square root solution in \eqref{eq:omkp1-quadratic-soln} gives $\omega_{k+1} = 0$, which is incorrect. Further, for sufficiently large values of $\omega_k$ and $\omega_{k+1}$, it can be seen from \eqref{eq:omkp1-zero-dyn} that
\begin{equation*} \label{eq:omkp1-approx}
    \omega_{k+1} \approx \frac{\Sp \omega_k}{\Sm}
\end{equation*}

\noindent which is also obtained on selecting the positive square root solution in \eqref{eq:omkp1-quadratic-soln} and approximating the expression inside the square root by $1$ for sufficiently large $\omega_k$. Thus, we choose the positive square root solution in \eqref{eq:omkp1-quadratic-soln}\footnote{Simulations demonstrate that this choice provides a viable solution for a large range of $\omega_k$ values.}.

If $\phi \in \{-\pi/2, \pi/2\}$, \eqref{eq:omkp1-quadratic-soln} takes the form
\begin{equation} \label{eq:omkp1-quadratic-soln-piby2}
    \omega_{k+1} = \frac{\omega_k}{2} + \frac{P \sin\theta_k}{2 \omega_k}
    \pm \frac{\omega_k}{2}
    \sqrt{1 +  \frac{6P \sin\theta_k}{\omega_k^2} + \frac{P^2 \sin^2\theta_k}{\omega_k^4}}
\end{equation}

\noindent which is a solution to \eqref{eq:omkp1-zero-dyn-piby2}. Again, if $\theta_k \!\!\mod 2\pi \in \{0, \pi\}$, the correct solution is obtained on choosing the positive square root solution in \eqref{eq:omkp1-quadratic-soln-piby2}, since it follows from \eqref{eq:omkp1-zero-dyn-piby2} that $\omega_{k+1} = \omega_k$. It can similarly be shown that the positive square root solution in \eqref{eq:omkp1-quadratic-soln-piby2} is the viable choice for a large range of $\omega_k$ values.

\subsection{Approximate Invariant} \label{appendix-invariant}

The continuous-time zero dynamics \eqref{eq:lim-zero-dyn} is known to have the following integral of motion \cite{perram_explicit_2003, khandelwal_propeller_2025}, which is invariant along solutions
\begin{equation} \label{eq:lim-invariant}
    E = e^{-2 \theta \cot\phi} \left[ \frac{1}{2}\dot\theta^2 - \frac{g (2 \sin\theta \cot\phi + \cos\theta)}{R\sin\phi (4\cot^2\phi + 1)} \right]
\end{equation}

Based on the integral of motion above, we propose the following:
\begin{proposition}
The following quantity is \emph{approximately} invariant\footnote{We have been unable to find an exact invariant along solutions to the discrete zero dynamics in \eqref{eq:zerodyn} and consider it to be an open problem.} along solutions of \eqref{eq:zerodyn}.
\begin{equation} \label{eq:invariant}
\begin{split}
    &\bar E_k = e^{-2 \cot\phi L_e \left(\theta_k - \delths/2\right)} \left[ \frac{1}{2}\omega_k^2 \right. \\
    &\left. - \frac{g L_t \left\{2 \cot\phi \sin\left(\theta_k - \dfrac{\delths}{2}\right) + \cos\left(\theta_k - \dfrac{\delths}{2}\right) \right\}}{R\sin\phi (4\cot^2\phi + 1)} \right]
\end{split}
\end{equation}

\noindent where
\begin{align*}
    L_e &\triangleq \frac{\sin(\delths/2)}{(\delths/2)} \frac{1}{\cos(\delths/2)} \\
    L_t &\triangleq \frac{(\delths/2)^2}{\sin^2(\delths/2)} \frac{1}{\cos(\delths/2)}
\end{align*}
\end{proposition}

\noindent It can be seen that $\lim_{\delths \to 0} L_e = 1$ and $\lim_{\delths \to 0} L_t = 1$; therefore, using \eqref{eq:lim-theta} in \eqref{eq:invariant}, we obtain
\begin{equation}
    \lim_{\delths \to 0} \bar E_k = E
\end{equation}

\noindent The approximate invariance of the quantity $\bar E_k$ is verified numerically. The variations in $\bar E_k$ are smaller for larger values of $\omega_k$ and smaller values of $\delths$.

If $\phi \in \{-\pi/2, \pi/2\}$, $\cot\phi = 0$, and \eqref{eq:invariant} reduces to
\begin{equation} \label{eq:invariant-piby2}
    \bar E_k = 
    \begin{cases}
        &\frac{1}{2}\omega_k^2 + \dfrac{g L_t \cos\left(\theta_k - \frac{\delths}{2}\right)}{R}, \quad \phi = -\pi/2 \\
        &\frac{1}{2}\omega_k^2 - \dfrac{g L_t \cos\left(\theta_k - \frac{\delths}{2}\right)}{R}, \quad \phi =  \pi/2
    \end{cases}
\end{equation}

\noindent If $\bar E_k$ were an \emph{exact} invariant, the above expression implies
\begin{itemize}
    \item periodic behavior of $\omega_k$,
    \item extrema of $\omega_k$ at $\theta_k \!\!\mod 2\pi = \delths/2$ and $\theta_k \!\!\mod 2\pi = \pi + \delths/2$, and 
    \item $\omega_{k+1} = \omega_k$ when $\theta_k \!\!\mod 2\pi \in \{0, \pi\}$, which agrees with \eqref{eq:omkp1-zero-dyn-piby2}.
\end{itemize}

\begin{remark}
    Zero dynamics induced by VHCs have an integral of motion which provides the nature of solutions on the constraint manifold. 
    Similarly, an exact invariant for the DZD induced by DVHCs would provide the nature of solutions for the solution set described in Remark \ref{rem:solution-set}. Here, we have an approximate invariant, which provides an estimate of the nature of solutions.
\end{remark}

\subsection{Nature of Solutions} \label{sec44}

It can be seen from \eqref{eq:lim-invariant} that periodic orbits for \eqref{eq:lim-zero-dyn} are possible only when $\phi \in \{-\pi/2, \pi/2\}$ \cite{khandelwal_propeller_2025}. 
It is therefore reasonable to expect that \eqref{eq:zerodyn} also permits periodic orbits only when $\phi \in \{-\pi/2, \pi/2\}$, \emph{i.e.}, when the zero dynamics is expressed by \eqref{eq:zerodyn-special} and the approximate invariant by \eqref{eq:invariant-piby2}.
The following results on the nature of the solutions to \eqref{eq:zerodyn-special} have partly been presented in \cite{khandelwal_two-dimensional_2024}.
We first observe that periodic solutions to \eqref{eq:hybrid-theta-delta} are only permitted if $\delths$ is a rational submultiple of $2\pi$, \emph{i.e.}, $\delths = (p/q) 2\pi$, where $p, q \in Z^+$. If $\gcd(p, q) = 1$, $\theta \!\!\mod 2\pi$ is $q$-periodic since
\begin{equation} \label{eq:th-rational-periodicity}
    \theta_{k+q} = \theta_k + 2\pi p = \theta_k \mod 2\pi
\end{equation}

\noindent and $\theta \!\!\mod 2\pi$ returns to its original value after $p$ revolutions around the circle. In particular, if $p = 1$ and $q = N$, $\delths = 2\pi/N$; $\theta \!\!\mod 2\pi$ is periodic with period $N$, and $\theta \!\!\mod 2\pi$ returns to its original value after a single revolution.

For the special case when $\delths$ is given by 
\begin{equation} \label{eq:delths}
   \delths = \frac{2\pi}{N}, \quad N \in Z^+, \quad N \geq 3
\end{equation}
we prove that periodic solutions to \eqref{eq:zerodyn-special} are guaranteed under specific conditions, beginning with the assumption
\begin{assumption} \label{asm:omk-min-periodic}
For any initial condition $(\theta_1, \omega_1)$,
\begin{equation*} 
    \omega_{k, \mathrm{min}}^2 > |P|, \qquad  \omega_{k, \mathrm{min}} \triangleq \min_{k} \omega_k
    \end{equation*}
\end{assumption}

\begin{lemma} \label{lem:zero}
The dynamical system in \eqref{eq:zerodyn-special} is periodic with period $N$ for $\theta_1 = 0$ and initial conditions satisfying Assumption \ref{asm:omk-min-periodic} when $\delths$ is given by \eqref{eq:delths}.
\end{lemma}

\begin{proof}
It follows from \eqref{eq:hybrid-theta-delta} and \eqref{eq:delths} that
\begin{equation*}
    \theta_{N+1} = \theta_1 + 2\pi = \theta_1 \mod 2\pi
\end{equation*}
which establishes periodicity of $\theta$. For any value of $N$ in \eqref{eq:delths}, we can express $\theta_k$ as follows
\begin{equation} \label{eq:thk-from-pi}
    \theta_k = (k-1) \delths = \pi + (2k - N - 2) \frac{\delths}{2}
\end{equation}

\noindent from which it follows that
\begin{equation} \label{eq:sin-from-pi}
    \sin\theta_k = - \sin\left[(2k - N - 2) \frac{\delths}{2}\right]
\end{equation}

\noindent Using \eqref{eq:sin-from-pi}, the dynamics \eqref{eq:omkp1-zero-dyn-piby2} may be rewritten as
\begin{equation} \label{eq:omkp1-zero-dyn-from-pi}
    \omega_{k+1} - \omega_k = - P \sin\left[(2k - N - 2) \frac{\delths}{2}\right] \left[ \frac{1}{\omega_k} + \frac{1}{\omega_{k+1}} \right]
\end{equation}

\noindent Substituting $k = 1$, we get
\begin{equation} \label{eq:om2-eq-om1}
    \omega_{2} = \omega_1
\end{equation}

\noindent We prove periodicity of $\omega$ separately for even and odd values of $N$.

\paragraph{Even $N$}

Substituting $k = \frac{N}{2}+1$ in \eqref{eq:thk-from-pi} and \eqref{eq:omkp1-zero-dyn-from-pi}, we get $\theta_{\frac{N}{2}+1} = \pi$ and $\omega_{\frac{N}{2}+2} - \omega_{\frac{N}{2}+1} = 0$, from which it follows that
\begin{equation} \label{eq:equal-pi}
    \omega_{\frac{N}{2}+2} = \omega_{\frac{N}{2}+1}
\end{equation}

\noindent Choosing $k = \left(\frac{N}{2}+1\right) - 1$ and $k = \left(\frac{N}{2}+1\right) + 1$ in \eqref{eq:omkp1-zero-dyn-from-pi}, we obtain
\begin{align} 
    \omega_{\frac{N}{2}+1} - \omega_{\frac{N}{2}} &= P \sin(\delths) \left[ \frac{1}{\omega_{\frac{N}{2}}} + \frac{1}{\omega_{\frac{N}{2}+1}} \right] \label{eq:even-Nby2} \\
    \omega_{\frac{N}{2}+3} - \omega_{\frac{N}{2}+2} &= - P \sin(\delths) \left[ \frac{1}{\omega_{\frac{N}{2}+2}} + \frac{1}{\omega_{\frac{N}{2}+3}} \right] \label{eq:even-Nby2p2}
\end{align}

\noindent Adding the above two equations, we obtain
\begin{equation}
\begin{split}
    &\omega_{\frac{N}{2}+3} - \omega_{\frac{N}{2}+2} + \omega_{\frac{N}{2}+1} - \omega_{\frac{N}{2}} \\&= P \sin(\delths) \left[ \frac{1}{\omega_{\frac{N}{2}}} + \frac{1}{\omega_{\frac{N}{2}+1}} - \frac{1}{\omega_{\frac{N}{2}+2}} - \frac{1}{\omega_{\frac{N}{2}+3}} \right]
\end{split}
\end{equation}

\noindent Using \eqref{eq:equal-pi}, the above equation simplifies to
\begin{equation}
    \omega_{\frac{N}{2}+3} - \omega_{\frac{N}{2}} = P \sin(\delths) \left[ \frac{\omega_{\frac{N}{2}+3} - \omega_{\frac{N}{2}}}{\omega_{\frac{N}{2}} \omega_{\frac{N}{2}+3}} \right]
\end{equation}
From Assumption \ref{asm:omk-min-periodic}, it follows
\begin{equation} \label{eq:equal-delths}
    \omega_{\frac{N}{2}+3} = \omega_{\frac{N}{2}}
\end{equation}

\noindent Similarly, for $k = \left(\frac{N}{2}+1\right) - 2$ and $k = \left(\frac{N}{2}+1\right) + 2$ in \eqref{eq:omkp1-zero-dyn-from-pi}, we obtain
\begin{align} 
    \omega_{\frac{N}{2}} - \omega_{\frac{N}{2}-1} &= P \sin(2\delths) \left[ \frac{1}{\omega_{\frac{N}{2}-1}} + \frac{1}{\omega_{\frac{N}{2}}} \right] \label{eq:even-Nby2m1} \\
    \omega_{\frac{N}{2}+4} - \omega_{\frac{N}{2}+3} &= - P \sin(2\delths) \left[ \frac{1}{\omega_{\frac{N}{2}+3}} + \frac{1}{\omega_{\frac{N}{2}+4}} \right] \label{eq:even-Nby2p3}
\end{align}

\noindent Adding the above two equations, we obtain
\begin{equation}
\begin{split}
    &\omega_{\frac{N}{2}+4} - \omega_{\frac{N}{2}+3} + \omega_{\frac{N}{2}} - \omega_{\frac{N}{2}-1} \\&= P \sin(2\delths) \left[ \frac{1}{\omega_{\frac{N}{2}-1}} + \frac{1}{\omega_{\frac{N}{2}}} - \frac{1}{\omega_{\frac{N}{2}+3}} - \frac{1}{\omega_{\frac{N}{2}+4}} \right]
\end{split}
\end{equation}

\noindent Using \eqref{eq:equal-delths}, the above equation simplifies to
\begin{equation}
    \omega_{\frac{N}{2}+4} - \omega_{\frac{N}{2}-1} = P \sin(2\delths) \left[ \frac{\omega_{\frac{N}{2}+4} - \omega_{\frac{N}{2}-1}}{\omega_{\frac{N}{2}-1} \omega_{\frac{N}{2}+4}} \right]
\end{equation}
From Assumption \ref{asm:omk-min-periodic}, it follows
\begin{equation} \label{eq:equal-2delths}
    \omega_{\frac{N}{2}+4} = \omega_{\frac{N}{2}-1}
\end{equation}

\noindent This process can be repeated for all pairs of values of $k = \left(\frac{N}{2}+1\right) - n$ and $k = \left(\frac{N}{2}+1\right) + n$, $n = 0, 1, 2, \dots, \left(\frac{N}{2}-1\right)$ to establish that 
\begin{equation} \label{eq:om-pairs-even}
    \omega_{\frac{N}{2}+n+2} = \omega_{\frac{N}{2}-n+1}, \quad n = 0, 1, 2, \dots, \left(\frac{N}{2}-1\right)
\end{equation}

\noindent In particular, for $n = \left(\frac{N}{2}-1\right)$, we get
\begin{equation}
    \omega_{N+1} = \omega_2
\end{equation}

\noindent Since $\omega_2 = \omega_1$ from \eqref{eq:om2-eq-om1}, we have
\begin{equation} \label{eq:om-periodicity-even}
    \omega_{N+1} = \omega_1
\end{equation}

\paragraph{Odd $N$}

Substituting $k = \frac{N+1}{2}$ and $k = \frac{N+3}{2}$ in \eqref{eq:omkp1-zero-dyn-from-pi}, we obtain
\begin{align} 
    \omega_{\frac{N+3}{2}} - \omega_{\frac{N+1}{2}} &= P \sin\left(\frac{\delths}{2}\right) \left[ \frac{1}{\omega_{\frac{N+1}{2}}} + \frac{1}{\omega_{\frac{N+3}{2}}} \right] \label{eq:odd-Np1by2} \\
    \omega_{\frac{N+5}{2}} - \omega_{\frac{N+3}{2}} &= - P \sin\left(\frac{\delths}{2}\right) \left[ \frac{1}{\omega_{\frac{N+3}{2}}} + \frac{1}{\omega_{\frac{N+5}{2}}} \right] \label{eq:odd-Np3by2}
\end{align}

\noindent Adding the above two equations, we obtain
\begin{equation}
    \omega_{\frac{N+5}{2}} - \omega_{\frac{N+1}{2}} = P \sin\left(\frac{\delths}{2}\right) \left[ \frac{1}{\omega_{\frac{N+1}{2}}} - \frac{1}{\omega_{\frac{N+5}{2}}} \right]
\end{equation}
which simplifies to
\begin{equation}
    \omega_{\frac{N+5}{2}} - \omega_{\frac{N+1}{2}} = P \sin\left(\frac{\delths}{2}\right) \left[ \frac{\omega_{\frac{N+5}{2}} - \omega_{\frac{N+1}{2}}}{\omega_{\frac{N+1}{2}}\omega_{\frac{N+5}{2}}} \right]
\end{equation}
From Assumption \ref{asm:omk-min-periodic}, it follows
\begin{equation} \label{eq:equal-delthsby2}
    \omega_{\frac{N+5}{2}} = \omega_{\frac{N+1}{2}}
\end{equation}

\noindent Similarly, choosing $k = \frac{N+1}{2}-1$ and $k = \frac{N+3}{2}+1$ in \eqref{eq:omkp1-zero-dyn-from-pi}, we obtain
\begin{align} 
    \omega_{\frac{N+1}{2}} - \omega_{\frac{N-1}{2}} &= P \sin\left(\frac{3\delths}{2}\right) \left[ \frac{1}{\omega_{\frac{N-1}{2}}} + \frac{1}{\omega_{\frac{N+1}{2}}} \right] \label{eq:odd-Nm1by2} \\
    \omega_{\frac{N+7}{2}} - \omega_{\frac{N+5}{2}} &= - P \sin\left(\frac{3\delths}{2}\right) \left[ \frac{1}{\omega_{\frac{N+5}{2}}} + \frac{1}{\omega_{\frac{N+7}{2}}} \right] \label{eq:odd-Np5by2}
\end{align}

\noindent Adding the above two equations, we obtain
\begin{equation}
\begin{split}
    &\omega_{\frac{N+7}{2}} - \omega_{\frac{N+5}{2}} + \omega_{\frac{N+1}{2}} - \omega_{\frac{N-1}{2}} \\&= P \sin\left(\frac{3\delths}{2}\right) \left[ \frac{1}{\omega_{\frac{N-1}{2}}} + \frac{1}{\omega_{\frac{N+1}{2}}} - \frac{1}{\omega_{\frac{N+5}{2}}} - \frac{1}{\omega_{\frac{N+7}{2}}} \right]
\end{split}
\end{equation}

\noindent Substituting \eqref{eq:equal-delthsby2} in the above equation and simplifying, we obtain
\begin{equation}
    \omega_{\frac{N+7}{2}} - \omega_{\frac{N-1}{2}} = P \sin\left(\frac{3\delths}{2}\right) \left[ \frac{\omega_{\frac{N+7}{2}} - \omega_{\frac{N-1}{2}}}{\omega_{\frac{N-1}{2}}\omega_{\frac{N+7}{2}}} \right]
\end{equation}
Using Assumption \ref{asm:omk-min-periodic}, we get
\begin{equation} \label{eq:equal-3delthsby2}
    \omega_{\frac{N+7}{2}} = \omega_{\frac{N-1}{2}}
\end{equation}

\noindent This process can be repeated for all pairs of values of $k = \frac{N+1}{2} - n$, $k = \frac{N+3}{2} + n$, $n = 0, 1, 2, \dots, \left(\frac{N-3}{2}\right)$ to establish that 
\begin{equation}  \label{eq:om-pairs-odd}
    \omega_{\frac{N+3}{2}+n+1} = \omega_{\frac{N+1}{2}-n}, \quad n = 0, 1, 2, \dots, \left(\frac{N-3}{2}\right)
\end{equation}

\noindent In particular, if $n = \left(\frac{N-3}{2}\right)$, the above equation gives
\begin{equation}
    \omega_{N+1} = \omega_2
\end{equation}

\noindent Since $\omega_2 = \omega_1$ from \eqref{eq:om2-eq-om1}, we have
\begin{equation} \label{eq:om-periodicity-odd}
    \omega_{N+1} = \omega_1
\end{equation}

\noindent This concludes the proof.
\end{proof}

\begin{lemma} \label{lem:delths}
The dynamical system in in \eqref{eq:zerodyn-special} is periodic with period $N$ for $\theta_1 = \delths/2$ and initial conditions satisfying Assumption \ref{asm:omk-min-periodic} when $\delths$ is given by \eqref{eq:delths}.
\end{lemma}
\begin{proof}
The proof is similar to the one for Lemma \ref{lem:zero}, and is omitted here for brevity.
\end{proof}

\begin{theorem} \label{thm:periodic}
The discrete zero dynamics in \eqref{eq:zerodyn-special} is periodic with period $N$ if $\theta_k \!\!\mod 2\pi \in \{0, \delths/2\}$ for some $k$, and the initial conditions satisfy Assumption \ref{asm:omk-min-periodic} when $\delths$ is given by \eqref{eq:delths}.
\end{theorem}
\begin{proof}
If $\theta_k \!\!\mod 2\pi \in \{0, \delths/2\}$ for some $k$, the pair $(\theta_k, \omega_k)$ may be treated as an equivalent initial condition $(\theta_1, \omega_1)$ for the dynamical system in \eqref{eq:zerodyn-special}. The proof now follows from Lemmas \ref{lem:zero} and \ref{lem:delths}, which assert periodicity of the dynamical system in \eqref{eq:zerodyn-special} for initial conditions with $\theta_1 = 0$ and $\theta_1 = \delths/2$ respectively as long as Assumption \ref{asm:omk-min-periodic} holds. 
\end{proof}

\begin{corollary} \label{cor:stable}
    The periodic orbits described in Theorem \ref{thm:periodic} are stable but not attractive.
\end{corollary}

\begin{proof}
    The proof follows directly from Theorem \ref{thm:periodic}. Let the periodic solution be given by the points $(\theta_1, \omega_1)$, $(\theta_2, \omega_2)$, \dots, $(\theta_N, \omega_N)$. If any periodic point $(\theta_k, \omega_k)$ is perturbed to $(\theta_k, \omega_k + \epsilon_k)$, the equivalent initial condition $(\theta_1, \omega_1 + \epsilon_1)$ can be found by decrementing $k$ while solving \eqref{eq:zerodyn-special}. Since the new initial condition satisfies Assumption \ref{asm:omk-min-periodic}, it results in a distinct periodic orbit with $(\theta_{N+1}, \omega_{N+1}) = (\theta_1, \omega_1 + \epsilon_1)$.
\end{proof}
 
It follows from Corollary \ref{cor:stable} that the eigenvalues of the Floquet Matrix \cite{ChaosBook}
\begin{equation} \label{eq:floquet-matrix}
    M = J_N \cdots J_2\, J_1, \quad
    J_k(\theta_k, \omega_k) \triangleq \begin{bmatrix}
        \dfrac{\partial\theta_{k+1}}{\partial\theta_k} & \dfrac{\partial\theta_{k+1}}{\partial\omega_k} \\
        \dfrac{\partial\omega_{k+1}}{\partial\theta_k} & \dfrac{\partial\omega_{k+1}}{\partial\omega_k}
    \end{bmatrix}
\end{equation}

\noindent have magnitude equal to $1$. This is verified numerically.

Under the condition that $\delths = (p/q) 2\pi$, $p, q \in Z^+$, we make the following additional observations about \eqref{eq:zerodyn-special} based on extensive numerical investigation:
\begin{enumerate}
    \item The solutions are periodic when $\theta_k \!\!\mod 2\pi \in \{0, \delths/2\}$ for some $k$. 
    \item If $\theta_k \!\!\mod 2\pi \notin \{0, \delths/2\}$ for any $k$, the solutions are not periodic, and a gradual drift in the values of $\omega_k$ associated with the same value of $\theta_k \!\!\mod 2\pi$ is observed.
    \begin{enumerate}
        \item The magnitude of the drift in $\omega_k$ is smaller if the values of $\omega_k$ are larger.
        \item The magnitude of the drift in $\omega_k$ is smaller if the value of $\delths$ is smaller, and the DZD more closely approximates the zero dynamics for the continuous-time system.
    \end{enumerate}
\end{enumerate}

\begin{remark}
    It follows from item 2 above that the value of $\bar E_k$ in \eqref{eq:invariant-piby2} drifts commensurate with the drift in $\omega_k$ when $\phi \in \{-\pi/2, \pi/2\}$.
\end{remark}

When $\phi \notin \{-\pi/2, \pi/2\}$, numerical investigation verifies that the DZD in \eqref{eq:zerodyn} does not exhibit periodic motion for any choice of $\delths$. The value of $\omega_k$ increases or decreases as the devil-stick completes successive rotations. This behavior agrees with the results for the continuous-time propeller motion \cite{khandelwal_propeller_2025}. 

\subsection{Simulation of Zero Dynamics} \label{sec45}

\begin{figure}[t]
    \centering
    \psfrag{A}{\hspace{-10pt} \footnotesize{$\theta$ (rad)}}
    \psfrag{B}{\hspace{-18pt} \footnotesize{$\omega$ (rad/s)}}
    \psfrag{T}{\footnotesize{$k$}}
    \psfrag{Q}{\footnotesize{$\bar E_k$}}
    \psfrag{E}{\footnotesize{$\bar E = $}}
    \psfrag{F}{\hspace{-10pt} \footnotesize{$I_k$ (Ns)}}
    \psfrag{R}{\hspace{-10pt} \footnotesize{$r_k$ (m)}}
    \psfrag{C}{$\rho_x$}
    \psfrag{D}{$\rho_y$}
    \psfrag{G}{$\drho_x$}
    \psfrag{L}{$\drho_y$}
    \psfrag{O}{$\mathcal{O}^*$}
    \includegraphics[width=\linewidth]{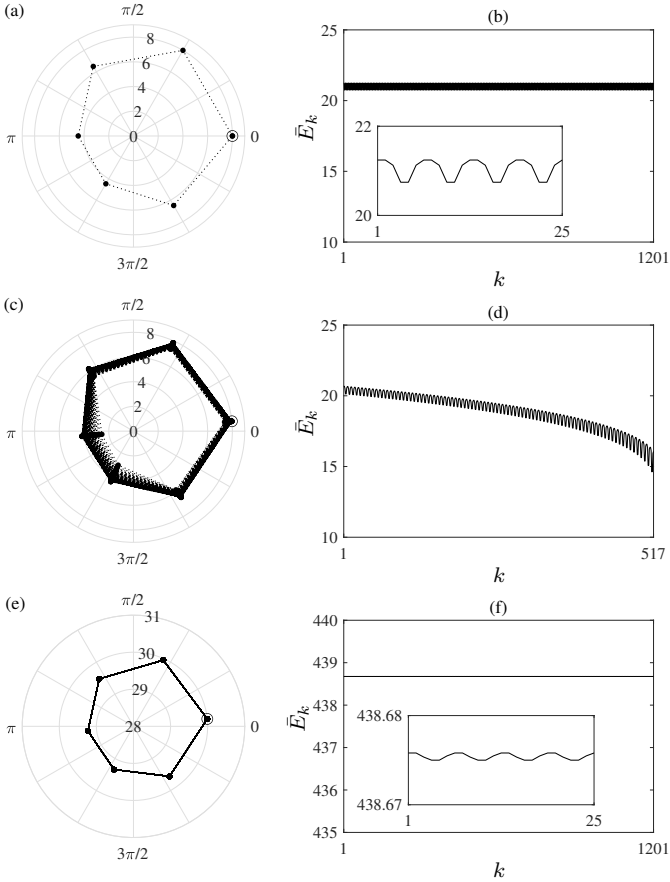}
    \caption{Simulation of the discrete zero dynamics of propeller motion using purely impulsive inputs for $\delths = 2\pi/6$, $R = 1$ and $\phi = \pi/2$: 
    (a) Polar plot showing periodic behavior and (b) Approximate invariant $\bar E_k$, both for initial condition $(\theta_1, \omega_1) = (0, 8)$ over 1200 simulation steps; 
    (c) Polar plot showing gradual decrease in $\omega$ and (d) Approximate invariant $\bar E_k$, both for initial condition $(\theta_1, \omega_1) = (0.1, 8)$ over 516 simulation steps, after which real solutions to \eqref{eq:omkp1-quadratic-soln-piby2} are not obtained;
    (e) Polar plot showing negligible drift in $\omega$ and (f) Approximate invariant $\bar E_k$, both for initial condition $(\theta_1, \omega_1) = (0.1, 30)$ over 1200 simulation steps.
    The insets in (b) and (f) show a magnified view of the fluctuations in $\bar E_k$ over 4 complete rotations, or 24 simulation steps.}
    \label{fig:zero-N6}
\end{figure}

\begin{figure}[h]
    \centering
    \psfrag{A}{\hspace{-10pt} \footnotesize{$\theta$ (rad)}}
    \psfrag{B}{\hspace{-18pt} \footnotesize{$\omega$ (rad/s)}}
    \psfrag{T}{\footnotesize{$k$}}
    \psfrag{Q}{\footnotesize{$\bar E_k$}}
    \psfrag{E}{\footnotesize{$\bar E = $}}
    \psfrag{F}{\hspace{-10pt} \footnotesize{$I_k$ (Ns)}}
    \psfrag{R}{\hspace{-10pt} \footnotesize{$r_k$ (m)}}
    \psfrag{C}{$\rho_x$}
    \psfrag{D}{$\rho_y$}
    \psfrag{G}{$\drho_x$}
    \psfrag{L}{$\drho_y$}
    \psfrag{O}{$\mathcal{O}^*$}
    \includegraphics[width=\linewidth]{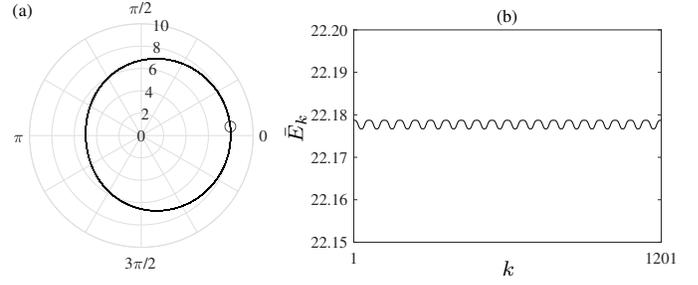}
    \caption{Simulation of the discrete zero dynamics of propeller motion using purely impulsive inputs for $\delths = 2\pi/60$, $R = 1$ and $\phi = \pi/2$ over 1200 simulation steps for initial condition $(\theta_1, \omega_1) = (0.1, 8)$: 
    (a) Polar plot and (b) Approximate invariant $\bar E_k$.}
    \label{fig:zero-limit}
\end{figure}

\begin{figure}[h]
    \centering
    \psfrag{A}{\hspace{-10pt} \footnotesize{$\theta$ (rad)}}
    \psfrag{B}{\hspace{-18pt} \footnotesize{$\omega$ (rad/s)}}
    \psfrag{T}{\footnotesize{$k$}}
    \psfrag{Q}{\footnotesize{$\bar E_k$}}
    \psfrag{E}{\footnotesize{$\bar E = $}}
    \psfrag{F}{\hspace{-10pt} \footnotesize{$I_k$ (Ns)}}
    \psfrag{R}{\hspace{-10pt} \footnotesize{$r_k$ (m)}}
    \psfrag{C}{$\rho_x$}
    \psfrag{D}{$\rho_y$}
    \psfrag{G}{$\drho_x$}
    \psfrag{L}{$\drho_y$}
    \psfrag{O}{$\mathcal{O}^*$}
    \includegraphics[width=\linewidth]{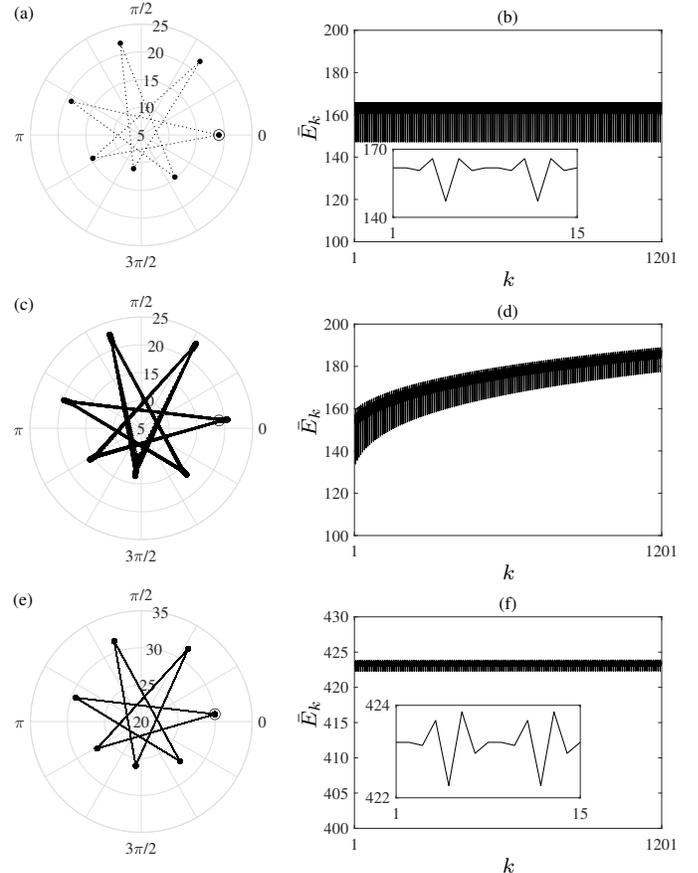}
    \caption{Simulation of the discrete zero dynamics of propeller motion using purely impulsive inputs for $\delths = (3/7) 2\pi$, $R = 1$ and $\phi = \pi/2$ over 1200 simulation steps: 
    (a) Polar plot showing periodic behavior and (b) Approximate invariant $\bar E_k$, both for initial condition $(\theta_1, \omega_1) = (0, 19)$; 
    (c) Polar plot showing gradual increase in $\omega$ and (d) Approximate invariant $\bar E_k$, both for initial condition $(\theta_1, \omega_1) = (0.1, 19)$;
    (e) Polar plot showing negligible drift in $\omega$ and (f) Approximate invariant $\bar E_k$, both for initial condition $(\theta_1, \omega_1) = (0.1, 30)$.
    The insets in (b) and (f) show a magnified view of the fluctuations in $\bar E_k$ over 6 complete rotations, or 14 simulation steps.}
    \label{fig:zero-rational}
\end{figure}

We present a numerical investigation of the zero dynamics in \eqref{eq:zerodyn-special} with the choices
\begin{equation*} 
    R = 1, \quad \phi = \pi/2
\end{equation*}

\noindent to support the claims in Section \ref{sec44} and demonstrate the dependence of the system dynamics on the choice of $\delths$ and initial conditions. The system is simulated over a large number of simulation steps to study the long-term behavior of the system.

\begin{figure}[t]
    \centering
    \psfrag{A}{\hspace{-10pt} \footnotesize{$\theta$ (rad)}}
    \psfrag{B}{\hspace{-18pt} \footnotesize{$\omega$ (rad/s)}}
    \psfrag{T}{\footnotesize{$k$}}
    \psfrag{Q}{\footnotesize{$\bar E_k$}}
    \psfrag{E}{\footnotesize{$\bar E = $}}
    \psfrag{F}{\hspace{-10pt} \footnotesize{$I_k$ (Ns)}}
    \psfrag{R}{\hspace{-10pt} \footnotesize{$r_k$ (m)}}
    \psfrag{C}{$\rho_x$}
    \psfrag{D}{$\rho_y$}
    \psfrag{G}{$\drho_x$}
    \psfrag{L}{$\drho_y$}
    \psfrag{O}{$\mathcal{O}^*$}
    \includegraphics[width=\linewidth]{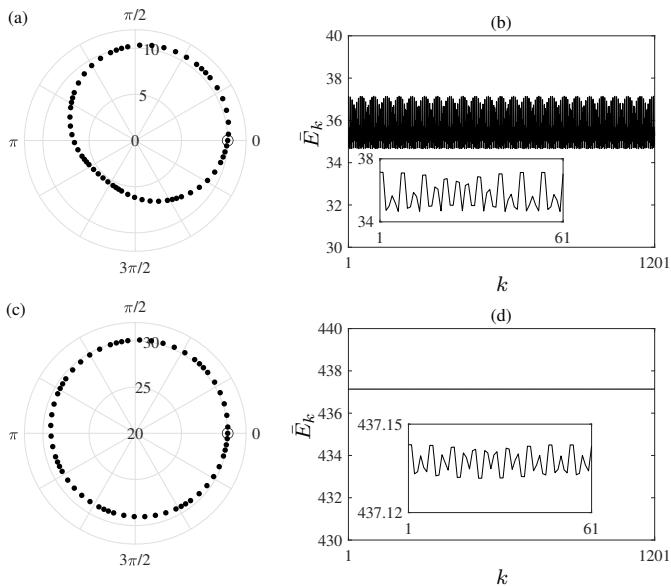}
    \caption{Simulation of the discrete zero dynamics of propeller motion using purely impulsive inputs for $\delths = (\sqrt{2}/5) 2\pi$, $R = 1$ and $\phi = \pi/2$ over 1200 simulation steps: 
    (a) Polar plot shown for 60 steps and (b) Approximate invariant $\bar E_k$, both for initial condition $(\theta_1, \omega_1) = (0, 10)$; 
    (c) Polar plot shown for 60 steps and (d) Approximate invariant $\bar E_k$, both for initial condition $(\theta_1, \omega_1) = (0, 30)$.
    The insets in (b) and (d) show a magnified view of the fluctuations in $\bar E_k$ over 60 simulation steps.}
    \label{fig:zero-irrational}
\end{figure}

\subsubsection{$\delths$: Integer submultiple of $2\pi$}
We first consider $\delths$ to have the form in \eqref{eq:delths} with $N = 6$. 
From the initial condition $(\theta_1, \omega_1) = (0, 8)$, Fig.\ref{fig:zero-N6}(a) shows the periodic behavior of the system over 1200 simulation steps. The approximate invariant $\bar E_k$ from \eqref{eq:invariant-piby2} is shown in Fig.\ref{fig:zero-N6}(b); its value fluctuates periodically. Similar periodic behavior is observed if the value of $\theta_1$ is changed to $\delths/2$, or if the value of $\omega_1$ is changed ensuring Assumption \ref{asm:omk-min-periodic} is satisfied. The fluctuations in $\bar E_k$ are smaller for larger values of $\omega_k$, indicating that \eqref{eq:invariant-piby2} is more accurate for larger values of $\omega_k$.

For the initial condition $(\theta_1, \omega_1) = (0.1, 8)$, Fig.\ref{fig:zero-N6}(c) shows a gradual `drift' in $\omega$. The drift in $\omega_k$ from $\omega_1$, when $\theta_k \!\!\mod 2\pi = 0.1$ is found to be $-7.48 \%$ over 516 simulation steps. The simulation window is reduced from 1200 since \eqref{eq:omkp1-quadratic-soln-piby2} ceases to permit solutions as $\omega_k$ reduces below a certain value. The approximate invariant $\bar E_k$ from \eqref{eq:invariant-piby2}, shown in Fig.\ref{fig:zero-N6}(d), is no longer periodic and reduces in value commensurate with the drift in $\omega_k$.

For the initial condition $(\theta_1, \omega_1) = (0.1, 30)$, Fig.\ref{fig:zero-N6}(e) shows that the behavior of the system appears to be very nearly periodic. The drift in $\omega$ is barely noticeable, with the drift in $\omega_k$ from $\omega_1$, when $\theta_k \!\!\mod 2\pi = 0.1$ found to be $-5.05 \times 10^{-5} \%$ over 1200 simulation steps. The value of $\bar E_k$ from \eqref{eq:invariant-piby2}, shown in Fig.\ref{fig:zero-N6}(f), fluctuates and drifts minimally owing to the large value of $\omega_1$. The drift in $\omega$ is also observed for other choices of $\theta_1$ for which $\theta_k \!\!\mod 2\pi \notin \{0, \delths/2\}$ for any $k$, with the drift being smaller for larger values of $\omega_k$.

For a large value of $N$, $N = 60$ in \eqref{eq:delths}, the system more closely approximates the continuous-time system. It can be seen from Fig.\ref{fig:zero-limit}(a) that the behavior of the system is apparently periodic from the initial condition $(\theta_1, \omega_1) = (0.1, 8)$, with a negligibly small drift in $\omega_k$, even though $\theta_k \!\!\mod 2\pi \notin \{0, \delths/2\}$ for any $k$ and $\omega_1$ is small. The fluctuation and drift in $\bar E_k$ from \eqref{eq:invariant-piby2}, shown in Fig.\ref{fig:zero-limit}(b) is also small. This contrasts the behavior seen in Fig.\ref{fig:zero-N6}(c)-(d), for the same initial conditions but a smaller value of $N$.

\subsubsection{$\delths$: Rational submultiple of $2\pi$}
We now consider $\delths = (3/7) 2\pi$, for which $\theta_{k+7} = \theta_k$ after three complete rotations of the devil-stick. From the initial condition $(\theta_1, \omega_1) = (0, 19)$, it can be seen from Fig.\ref{fig:zero-rational}(a) that the system shows periodic behavior over 1200 simulation steps. The value of $\bar E_k$ from \eqref{eq:invariant-piby2}, shown in Fig.\ref{fig:zero-rational}(b) fluctuates periodically. Periodic behavior is also observed if $\theta_1$ is changed to $\delths/2$, or if the value of $\omega_1$ is altered ensuring that \eqref{eq:omkp1-quadratic-soln-piby2} permits solutions.

For the initial condition $(\theta_1, \omega_1) = (0.1, 19)$, Fig.\ref{fig:zero-rational}(c) shows a drift in $\omega_k$ from $\omega_1$ by $8.41 \%$ over 1200 simulation steps. The value of $\bar E_k$ from \eqref{eq:invariant-piby2} is shown in Fig.\ref{fig:zero-rational}(d) and increases gradually.

For the initial condition $(\theta_1, \omega_1) = (0.1, 30)$, Fig.\ref{fig:zero-rational}(e) shows that the magnitude of the drift $\omega_k$ from $\omega_1$ is much smaller, $5.54 \times 10^{-3} \%$ over 1200 simulation steps. The fluctuation and drift in $\bar E_k$ from \eqref{eq:invariant-piby2}, shown in Fig.\ref{fig:zero-rational}(f), is also smaller.

\subsubsection{$\delths$: Irrational submultiple of $2\pi$}
If $\delths = (\sqrt{2}/5) 2\pi$, an irrational submultiple of $2\pi$, it is known that periodic motion cannot occur since \eqref{eq:hybrid-theta-delta} represents an irrational rotation. Over time, however, the iterations $\theta_k \!\!\mod 2\pi$ densely occupy points in $S^1$. Simulations of the system over 1200 steps with $\theta_1 = 0$ and $\omega_1 = 10$ and $30$ are shown in Fig.\ref{fig:zero-irrational}. Since $\theta_k$ values will densely fill the circle, the qualitative behavior of the system will be the same regardless of $\theta_1$, and different choices of $\theta_1$ need not be considered. It is seen for both choices of $\omega_1$ that the maxima and minima of $\omega_k$ also occur at $\theta_k \!\!\mod 2\pi \approx \delths/2$ and $\theta_k \!\!\mod 2\pi \approx \pi + \delths/2$ respectively, which is consistent with the prediction from \eqref{eq:invariant-piby2}. 
It can be seen from Fig.\ref{fig:zero-irrational}(b) and \ref{fig:zero-irrational}(d) that the fluctuation in the value of $\bar E_k$ from \eqref{eq:invariant-piby2} reduces as the values of $\omega_k$ increase.

\balance
\bibliographystyle{IEEEtran}      
\bibliography{ref}   

\end{document}